\documentclass[a4paper,11pt]{amsart}
\usepackage{amssymb,amsfonts,amsxtra,amscd,mathrsfs,color}
\usepackage[all]{xy}
\usepackage{fullpage}
\usepackage{graphicx}
\usepackage[colorlinks=true,linkcolor=blue,citecolor=red]{hyperref}
\theoremstyle{plain}
\newtheorem{theorem}{Theorem}[section]
\newtheorem{lemma}[theorem]{Lemma}
\newtheorem{cor}[theorem]{Corollary}
\newtheorem{prop}[theorem]{Proposition}
\theoremstyle{definition}
\newtheorem{defi}[theorem]{Definition}
\newtheorem{example}[theorem]{Example}
\theoremstyle{remark}
\newtheorem{rem}[theorem]{Remark}
\numberwithin{equation}{section}
\newcommand{\gf}{\ensuremath{\mathbb{K}}}
\newcommand{\Pois}[1]{\ensuremath{\mathrm{P}\!\left[#1\right]}}
\newcommand{\BV}[1]{\ensuremath{\mathrm{BV}\!\left[#1\right]}}
\newcommand{\NCHam}[1]{\ensuremath{\mathrm{H}\!\left[#1\right]}}
\newcommand{\NCHamP}[1]{\ensuremath{\mathrm{H}_+\!\left[#1\right]}}
\newcommand{\NCBV}[1]{\ensuremath{\mathrm{BV}^{\mathrm{nc}}_{\nu}\!\left[#1\right]}}
\newcommand{\NCBVg}[1]{\ensuremath{\mathrm{BV}^{\mathrm{nc}}_{\gamma,\nu}\!\left[#1\right]}}
\newcommand{\sym}[1]{\ensuremath{S}\!\left(#1\right)}
\newcommand{\mor}{\ensuremath{\mathcal{M}}}
\newcommand{\surf}{\ensuremath{\mathcal{S}}}
\newcommand{\chord}[1]{\ensuremath{\mathscr{C}\!\left(#1\right)}}
\newcommand{\var}[1]{\ensuremath{\mathrm{Var}\left(#1\right)}}
\newcommand{\cov}[2]{\ensuremath{\mathrm{Cov}\left(#1,#2\right)}}
\newcommand{\corr}[2]{\ensuremath{\rho}\left(#1,#2\right)}
\newcommand{\gl}[2]{\ensuremath{\mathfrak{gl}_{#1}(#2)}}
\newcommand{\GL}[2]{\ensuremath{GL_{#1}(#2)}}
\newcommand{\mat}[2]{\ensuremath{\mathrm{M}_{#1}(#2)}}
\newcommand{\her}[1]{\ensuremath{\mathfrak{h}_{#1}}}
\newcommand{\tr}{\ensuremath{\mathrm{Tr}}}
\newcommand{\symg}[1]{\ensuremath{\mathbb{S}_{#1}}}
\newcommand{\cycg}[1]{\ensuremath{\mathbb{Z}/{#1}\mathbb{Z}}}
\newcommand{\innprod}{\ensuremath{\langle -,- \rangle}}
\newcommand{\twodim}{\ensuremath{\mathcal{A}}}
\DeclareMathOperator{\id}{id}

\begin{document}
\title{A homological approach to the Gaussian Unitary Ensemble}
\author{Owen Gwilliam}
\address{Department of Mathematics and Statistics, Lederle Graduate Research Tower 1623D, University of Massachusetts Amherst, 710 N. Pleasant Street, Amherst, MA 01003-9305. USA.} \email{owen.gwilliam@gmail.com}
\author{Alastair Hamilton}
\address{Department of Mathematics and Statistics, Texas Tech University, Lubbock, TX 79409-1042. USA.} \email{alastair.hamilton@ttu.edu}
\author{Mahmoud Zeinalian}
\address{Lehman College, Department of Mathematics, 250 Bedford Park Blvd W, Bronx, NY 10468. USA.} \email{mahmoud.zeinalian@lehman.cuny.edu}
\begin{abstract}
We study the Gaussian Unitary Ensemble (GUE) using noncommutative geometry and the homological framework of the Batalin-Vilkovisky (BV) formalism.
Coefficients of the correlation functions in the GUE with respect to the rank $N$ are described in terms of ribbon graph Feynman diagrams that then lead to a counting problem for the corresponding surfaces.
The canonical relations provided by this homological setup determine a recurrence relation for these correlation functions. Using this recurrence relation and properties of the Catalan numbers,
we determine the leading order behavior of the correlation functions with respect to the rank~$N$. As an application, we prove a generalization of Wigner's semicircle law and compute all the large $N$ statistical correlations for the family of random variables in the GUE defined by multi-trace functions.
\end{abstract}


\makeatletter
\@namedef{subjclassname@2020}{%
  \textup{2020} Mathematics Subject Classification}
\makeatother

\maketitle

\section{Introduction}

In order to model complicated quantum mechanical systems, Wigner hit upon the clever idea of studying random matrices \cite{wignermatrix}. The most basic version involves placing a Gaussian measure
\[ \mu_N := \frac{1}{Z_N} e^{-\frac{1}{2}\tr{(X^2)}} \mathrm{d} X \]
on the space of $N\times N$ Hermitian matrices $\her{N}$, where $Z_N$ normalizes the measure to have mass one. In the large $N$ limit, he found that the expected values of trace functions such as
\[ \int_{\her{N}} \tr(X^{36}) \mu_N \]
exhibited remarkable behavior: to leading order in $N$, the expected value reduces to a simple one-dimensional integral.
For the function above,
\[ \lim_{N\to\infty}\frac{1}{N}\int_{\her{N}}\tr\left(\left(\frac{X}{\sqrt{N}}\right)^{36}\right)\mu_N
= \frac{1}{2\pi}\int_{-2}^2 x^{36}\sqrt{4-x^2}\,\mathrm{d}x; \]
an example of his semicircle law \cite{semicircle}. Out of Wigner's work grew an area of mathematics where physicists, probabilists, operator algebraists, and others meet.

One goal of this paper is to show how homological algebra offers a novel perspective that clarifies, in particular, how and why these amazing results in probability theory are related to noncommutative geometry -- specifically the noncommutative symplectic geometry of Kontsevich \cite{kontsympgeom} -- and hence to the topology of moduli spaces of Riemann surfaces.
From the work of Harer \cite{harer}, Mumford, Penner \cite{penner} and Thurston, it is known that the cohomology of the moduli space of Riemann surfaces may be described in terms of an orbi-cell complex generated by ribbon graphs.
By the results of \cite{kontsympgeom} and \cite{hamdgla}, the homological aspects of this complex may be recast in terms of the Batalin-Vilkovisky (BV) formalism and noncommutative geometry,
and it is this perspective that we use in this paper.

The BV formalism is a homological way to encode the idea of the path integral in general;
this framework has had great success in the study of gauge and string field theories.
As the Gaussian Unitary Ensemble (GUE) can be interpreted as $0$-dimensional quantum field theory, it is natural (at least to BV enthusiasts) to try it in this context. One appealing aspect of the BV approach is that it shifts the emphasis onto the algebra of observables rather than on constructing a measure, and this shift means that the large $N$ limit is being explored in an algebraic setting. Thanks to the Loday-Quillen-Tsygan (LQT) Theorem \cite{lodayquillen, tsygan} this limit is well-understood for the classical field theory.

In the paper \cite{GGHZ} we explained how to quantize the LQT theorem, and we found a very simple differential graded algebra that encodes the GUE.
(Our treatment here is mostly self-contained; we write so that the reader does not need to understand the LQT theorem in general, but instead the explicit formulas relevant to the GUE that we present in the body of the paper.)
In this paper we focus on exploiting that presentation to find recurrence relations among multi-trace expected values and in particular to find novel expressions for the leading order behavior that go beyond the original results of Wigner and his semicircle law (see Section~\ref{sec: corr} and \ref{sec_largeN} respectively).
In brief, the noncommutative geometry encodes algebraically the combinatorics of ribbon graphs, allowing us to make efficient computations.

A quick consequence of this large $N$ asymptotic analysis, performed using homological algebra and an understanding of the properties of the Catalan numbers, is a generalization of the semicircle law to multi-trace functions.
This generalization can be understood as showing that trace functions are asymptotically free in the sense of Voiculescu's free probability theory \cite{voiculimitlaw}; see Section~\ref{sec: voic} for these results.
As another application, we compute the large $N$ statistical correlation coefficients for a certain infinite family of random variables in the GUE defined by multi-trace functions;
see Section~\ref{sec_largeNcorcoeff} for these results.

Perhaps the most conceptually useful aspect of our methods is that it directly connects ribbon graphs, and hence the topology of the moduli space of Riemann surfaces, to the integration-by-parts relations between multi-trace functions arising through the BV formalism. This connection arises because the LQT theorem relates cyclic cohomology -- loosely speaking, part of a closed string field theory -- to Lie algebra cohomology -- loosely speaking, part of a gauge theory; making this assertion precise is a central part of \cite{GGHZ}. There is a deep reservoir of results about how closed string field theory relates to the topology of Riemann surfaces, and here we use but a handful of those insights.

As a guide for the reader, note that Section~\ref{sec: rec} is a lightning review of the BV formalism and the quantum LQT theorem of \cite{GGHZ}. Section~\ref{sec: GUE} then reviews how the GUE appears as an example in this framework. In Section~\ref{sec: rib} we begin on new material, reviewing ribbon graphs and explaining how they arise in the setting of the GUE and the quantum LQT theorem. Finally, we get to the main results: Section~\ref{sec: corr} proves that the multi-trace correlation functions are polynomial functions of the rank $N$ with nonnegative integer coefficients, determines their degree and establishes a recurrence relation for them, while Section~\ref{sec_largeN} computes their leading order behavior with respect to the rank $N$.

\subsection{Notation and conventions}

Throughout the paper our convention will be to work with differential graded symplectic vector spaces. We also follow the convention of working with cohomologically graded objects; hence the suspension $\Sigma V$ of a graded vector space $V$ is defined by $\Sigma V^i := V^{i+1}$. We will assume that our symplectic vector spaces $V$ carry a symplectic form $\innprod$ of \emph{odd} degree. The differential $d$ on $V$ is required to be compatible with the inner product in the sense that
\begin{equation} \label{eqn_diffinnprod}
\langle dx,y \rangle + (-1)^x \langle x,dy \rangle = 0.
\end{equation}
We emphasize that, for the sake of brevity, we will refer to these spaces simply as \emph{symplectic vector spaces}; with the understanding that they carry a differential (possibly zero) and that the symplectic form has odd degree.

We will define the inverse form $\innprod^{-1}$ on the dual space $V^*$ by the commutative diagram
\begin{equation} \label{eqn_invinnprod}
\xymatrix{ & \gf \\ V\otimes V \ar[ur]^{\innprod} \ar[rr]^{D_l \otimes D_r} && V^*\otimes V^* \ar[ul]_{\innprod^{-1}}}
\end{equation}
where $D_l(y):=\langle y,- \rangle$ and $D_r(y):=\langle -,y \rangle$. Note that while the form $\innprod$ is skew-symmetric, the Koszul sign rule implies that the inverse form $\innprod^{-1}$ is symmetric. More generally, we will use the same formula for the inverse of \emph{any} nondegenerate bilinear form $\innprod$.

We denote the symmetric group by $\symg{n}$. We follow the convention that coinvariants are indicated by a subscript. The graded symmetric algebra on a graded vector space $V$ will be denoted by $\sym{V}$. Throughout the paper we work over a ground field $\gf$ of characteristic zero, usually $\gf=\mathbb{C}$. The cardinality of a finite set $X$ will be denoted by $|X|$.

\section{Recollections on the quantum LQT theorem}
\label{sec: rec}

In this section we will recall the basic framework of the Batalin-Vilkovisky formalism \cite{schwarz}, including its formulation in noncommutative geometry coming from the work of Kontsevich \cite{kontsympgeom}. We then recall from \cite{GGHZ} how the quantum LQT maps intertwine the commutative and noncommutative aspects of this framework.

\subsection{The commutative geometry of the Batalin-Vilkovisky formalism}

We begin by describing the classical setup for the Batalin-Vilkovisky formalism and its commutative geometry.

\begin{defi}
Given a symplectic graded vector space $V$ we define
\[ \Pois{V}:=\sym{V^*}=\bigoplus_{k=0}^\infty \big[(V^*)^{\otimes k}\big]_{\symg{k}}. \]
From the inverse form $\innprod^{-1}$ on $V^*$ we define a Poisson bracket $\{-,-\}$ on $\Pois{V}$ of odd degree by extending the inverse form on $V^*$ to $\Pois{V}$ using the Leibniz rule;
\begin{equation} \label{eqn_leibniz}
\{a,bc\} = \{a,b\}c + (-1)^{(a+1)b} b\{a,c\}.
\end{equation}
\end{defi}

This structure is sometimes referred to as a (differential graded) \emph{shifted Poisson algebra}. There is a natural quantization of this structure, in the Batalin-Vilkovisky sense, by turning on a differential.

\begin{defi} \label{def_canonicalBValgebra}
The \emph{BV-Laplacian} $\Delta$ on $\sym{V^*}$ is the unique differential operator satisfying
\begin{equation} \label{eqn_BVidentity}
\Delta(ab) = (\Delta a)b + (-1)^a a(\Delta b) +\{a,b\}
\end{equation}
for all $a, b \in \sym{V^*}$ and such that $\Delta v$ vanishes on all $v\in V^*$. These conditions ensure that $\Delta^2=0$. From \eqref{eqn_diffinnprod} it follows that $(d+\Delta)$ is still a differential and more generally that
\[ \BV{V} := \big(\sym{V^*},(d+\Delta), \ \cdot \ ,\{-,-\}\big) \]
is a \emph{Batalin-Vilkovisky algebra}, where we have used $\ \cdot \ $ to denote the commutative multiplication on the underlying graded-commutative algebra.
\end{defi}

\subsection{Noncommutative geometry in the Batalin-Vilkovisky formalism}

We now recall the formulation of the Batalin-Vilkovisky formalism in the framework of noncommutative geometry. This begins with a construction of Kontsevich \cite{kontsympgeom}, with further input from the work of Movshev \cite{movshevcobracket}.

\begin{defi} \label{def_hamliebialg}
Given a symplectic graded vector space $V$, we define
\[ \NCHam{V}:=\bigoplus_{k=0}^\infty \big[(V^*)^{\otimes k}\big]_{\cycg{k}}. \]

We define a Lie bracket of odd degree on $\NCHam{V}$ by the formula
\[ \{(a_1\cdots a_m),(b_1\cdots b_n)\} := \sum_{i=1}^m\sum_{j=1}^n \pm\langle a_i,b_j \rangle^{-1} (a_{i+1}\cdots a_m a_1 \cdots a_{i-1} b_{j+1}\cdots b_n b_1\cdots b_{j-1}), \]
where $a_1,\ldots a_m, b_1,\ldots b_n\in V^*$ and the sign is determined canonically by the Koszul sign rule. Note that when $m = 1 = n$,
\[ \{(a),(b)\} = \langle a,b \rangle^{-1}, \]
which is a constant.

A Lie cobracket
\[ \nabla:\NCHam{V}\to\left(\NCHam{V}\otimes\NCHam{V}\right)_{\symg{2}} \]
of odd degree may also be defined by the formula
\[ \nabla(a_1\cdots a_n):=\sum_{1\leq i < j\leq n}\pm\langle a_i,a_j \rangle^{-1} (a_{i+1}\cdots a_{j-1})\otimes(a_{j+1}\cdots a_n a_1 \cdots a_{i-1}). \]
Again, the sign is determined canonically by the Koszul sign rule. These structures turn $\NCHam{V}$ into a Lie bialgebra.
\end{defi}

Consider the subspace of $\NCHam{V}$ consisting of positive cyclic powers,
\[ \NCHamP{V}:=\bigoplus_{k=1}^\infty \big[(V^*)^{\otimes k}\big]_{\cycg{k}}, \]
and denote the generator of $(V^*)^{\otimes 0}$ inside $\NCHam{V}$ by $\nu$, which has degree zero. Then we may write
\begin{equation} \label{eqn_nudecomposition}
\NCHam{V} = \gf\nu\oplus\NCHamP{V}.
\end{equation}

It is easily checked using the above definitions that the canonical quotient map
\[ \sigma:\NCHam{V}\to\Pois{V} \]
(which sends $\nu$ to $1$) is a map of Lie algebras. The commutative algebra $\sym{\NCHam{V}}$ has the canonical structure of a shifted Poisson algebra in which the bracket on $\NCHam{V}$ is extended to $\sym{\NCHam{V}}$ using the Leibniz rule \eqref{eqn_leibniz}. The above map then extends to a map
\begin{equation} \label{eqn_Poissonextension}
\sigma:\sym{{\NCHam{V}}}\to\Pois{V}
\end{equation}
of shifted Poisson algebras in a unique way.

Note that using the decomposition \eqref{eqn_nudecomposition} we may write
\begin{displaymath}
\begin{split}
\sym{\NCHam{V}} &= \sym{\gf\nu\oplus\NCHamP{V}} \\
&= \sym{\gf\nu}\otimes\sym{\NCHamP{V}} = \gf[\nu]\otimes\sym{\NCHamP{V}},
\end{split}
\end{displaymath}
where we have identified the symmetric algebra on $\gf\nu$ with polynomials in $\nu$.

We may now define, following \cite{GGHZ}, the noncommutative counterparts of the Batalin-Vilkovisky algebra described by Definition \ref{def_canonicalBValgebra}.

\begin{defi}
Given a symplectic graded vector space $V$, set
\[ \NCBV{V}:=\sym{\NCHam{V}} = \gf[\nu]\otimes\sym{\NCHamP{V}}. \]
We extend the shifted Poisson algebra structure on $\sym{\NCHam{V}}$ described above to the structure of a Batalin-Vilkovisky algebra by defining a BV-Laplacian using the formula
\[ \Delta_\nu := \nabla+\delta, \]
where $\delta$ denotes the Chevalley-Eilenberg differential on $\sym{\NCHam{V}}$ determined by the Lie bracket on $\NCHam{V}$ and $\nabla$ denotes the cobracket on $\NCHam{V}$, which is extended to $\sym{\NCHam{V}}$ using the Leibniz rule. This provides us with the structure of a Batalin-Vilkovisky algebra:
\[ \NCBV{V} = \big( \sym{\NCHam{V}}, (d+\Delta_\nu), \ \cdot \ , \{-,-\} \big). \]
\end{defi}

It is a consequence of Equation \eqref{eqn_BVidentity} that the map \eqref{eqn_Poissonextension} yields a map of Batalin-Vilkovisky algebras -- that is to say, it commutes with the BV-Laplacians. The details are not difficult, and are spelled out in Proposition 4.3 of \cite{GGHZ}. We will denote this map of BV algebras by
\begin{equation} \label{eqn_noncomBVmapnu}
\sigma_\nu:\NCBV{V}\to\BV{V}.
\end{equation}

We may refine the picture described above by including another deformation parameter, bringing us closer to the construction described in~\cite{hamdgla}.

\begin{defi}
Given a symplectic graded vector space $V$, set
\[ \NCBVg{V}:=\gf[\gamma]\otimes\sym{\NCHam{V}} = \gf[\gamma,\nu]\otimes\sym{\NCHamP{V}}. \]
This also has the structure of a shifted Poisson algebra, where we extend those structures that we defined on $\sym{\NCHam{V}}$ above linearly with respect to the parameter $\gamma$ (which has degree zero). We equip it with the BV-Laplacian
\[ \Delta_{\gamma,\nu} := \nabla + \gamma\cdot\delta. \]
This defines a Batalin-Vilkovisky algebra
\[ \NCBVg{V} = \big(\gf[\gamma]\otimes\sym{\NCHam{V}}, (d+\Delta_{\gamma,\nu}), \ \cdot \ , \gamma\{-,-\}\big). \]
Note that in this case we must multiply the Poisson bracket $\{-,-\}$ by the parameter $\gamma$ to retain the structure of a Batalin-Vilkovisky algebra.
\end{defi}

\begin{rem}
Later, we will use the parameters $\gamma$ and $\nu$ to keep track of the genus and number of boundary components associated to a ribbon graph, cf. Section~\ref{sec: exp in NC}.
\end{rem}

It is clear that by sending the deformation parameter $\gamma$ to $1$ we get a map of Batalin-Vilkovisky algebras
\[ \xymatrix{\NCBVg{V} \ar[r]^{\gamma=1} & \NCBV{V}}. \]
Combining this map with the map \eqref{eqn_noncomBVmapnu} yields a map of Batalin-Vilkovisky algebras which we denote by
\begin{equation} \label{eqn_noncomBVmapgamma}
\sigma_{\gamma,\nu}:\NCBVg{V}\to\BV{V}.
\end{equation}

\subsection{The Morita map}

In this section we will introduce the map that appears in the formulation of Morita invariance in Hochschild cohomology. It is defined by taking the trace of a product of matrices. It is the compatibility of this map with the framework of the Batalin-Vilkovisky formalism that will allow us to make contact with quantities in random matrix theory.

Let $V$ be a symplectic vector space and consider the space of $N$-by-$N$ matrices with entries in $V$, which we denote
\[ \mat{N}{V} = V\otimes\mat{N}{\gf}. \]
It is also a (differential graded) symplectic vector space whose symplectic form of odd degree is defined by
\[ \langle x\otimes A, y\otimes B \rangle := \langle x,y\rangle \tr(AB). \]

Consider the multilinear maps;
\begin{equation} \label{eqn_tracemap}
t_k:\mat{N}{\gf}^{\otimes k}\to\gf
\end{equation}
where
\[
t_k(A_1,\ldots,A_k):=\tr(A_1\cdots A_k).
\]
These are used to define the Morita map as follows.

\begin{defi}
Let $V$ be a symplectic vector space and consider the map;
\[ \big(V^{\otimes k}\big)^* \to \big(V^{\otimes k}\big)^*\otimes\big(\mat{N}{\gf}^{\otimes k}\big)^* = \big(\mat{N}{V}^{\otimes k}\big)^*\]
sending $\xi_k$ to $\xi_k\otimes t_k$.
As the multilinear maps \eqref{eqn_tracemap} are cyclically symmetric,
the above determines a well-defined map
\begin{equation} \label{eqn_moritamap}
\mor:\NCHam{V}\to\NCHam{\mat{N}{V}}.
\end{equation}
\end{defi}

Note that since the identity matrix has trace $\tr(I_N)=N$ it follows that
\[ \mor(\nu) = N\nu. \]

The following result, Lemma 4.6 of \cite{GGHZ}, shows that this map respects the structures on $\NCHam{V}$ introduced by Definition~\ref{def_hamliebialg}.

\begin{prop}
The map \eqref{eqn_moritamap} is a morphism of Lie bialgebras.
\end{prop}

This result has the following immediate consequence.
The map \eqref{eqn_moritamap} extends to a map
\[ \mor: \sym{\NCHam{V}}\to \sym{\NCHam{\mat{N}{V}}} \]
of shifted Poisson algebras. In fact, it also provides a map of BV algebras, which we denote by
\[ \mor_\nu:\NCBV{V}\to\NCBV{\mat{N}{V}}. \]

Denote the $\gamma$-linear extension of $\mor_\nu$ to $\NCBVg{V}=\gf[\gamma]\otimes\sym{\NCHam{V}}$ by
\[ \mor_{\gamma,\nu}:\NCBVg{V}\to\NCBVg{\mat{N}{V}}. \]
It is also a map of BV algebras.

We can combine these maps with the morphisms \eqref{eqn_noncomBVmapnu} and \eqref{eqn_noncomBVmapgamma} of BV algebras. Note that since the multilinear maps \eqref{eqn_tracemap} are $\gl{N}{\gf}$-invariant, the combined maps will land in the $\gl{N}{\gf}$-invariants $\BV{\mat{N}{V}}^{\gl{N}{\gf}}$. All told, we have the following commutative diagram of Batalin-Vilkovisky algebras
\[\xymatrix{
\NCBVg{V} \ar[dr]^{\sigma_{\gamma,\nu}\circ\mor_{\gamma,\nu}} \ar[dd]_{\gamma=1} & \\
& *+[r]{\BV{\mat{N}{V}}^{\gl{N}{\gf}}\subset\BV{\mat{N}{V}}} \\ \NCBV{V} \ar[ur]_{\sigma_\nu\circ\mor_\nu} }\]

\begin{rem} \label{rem_qLQTmaps}
The diagonal maps in this diagram are what we will sometimes refer to as the ``quantized LQT maps.'' This terminology is carried over from \cite{GGHZ}, where they correspond to the maps in the Loday-Quillen-Tsygan Theorem in the special case of a vanishing algebra structure on the graded vector space~$V$.
\end{rem}

\section{Matrix integrals in the Gaussian Unitary Ensemble}
\label{sec: GUE}

In this section we recall from \cite{GGHZ} how, as an application of the constructions just reviewed, a two-dimensional symplectic vector space encodes Hermitian matrix integrals and hence describes certain correlation functions in the Gaussian Unitary Ensemble. We are then able to articulate a key construction of this paper, which will be the subject of the next section.

Let $\her{N}$ denote the real subspace of $\gl{N}{\mathbb{C}}$ consisting of all Hermitian matrices. We are interested in studying the $k$-trace correlation functions
\begin{equation} \label{eqn_corfunmultitrace}
I_{k_1,k_2,\ldots,k_n}^N:=\frac{\int_{\her{N}}\tr (X^{k_1})\tr (X^{k_2})\cdots\tr (X^{k_n})e^{-\frac{1}{2}\tr (X^2)}\mathrm{d} X}{\int_{\her{N}}e^{-\frac{1}{2}\tr (X^2)}\mathrm{d} X}.
\end{equation}
Note that while the integrals in both the numerator and the denominator of \eqref{eqn_corfunmultitrace} depend upon a choice of linear identification of $\her{N}$ with $\mathbb{R}^{N^2}$,
the ratio does not.

The symplectic vector space leading to these correlation functions is remarkably simple.

\begin{defi}
Let $\twodim$ be the two-dimensional complex vector space with generators $a$ and $b$ of degrees zero and one respectively.
Define the differential graded symplectic structure on $\twodim$ by
\[ da = b, \qquad \langle b,a \rangle = 1 = -\langle a,b \rangle. \]
\end{defi}

Note that $\twodim$ is acyclic, so that many complexes we construct from it have simple cohomology. What will be of interest later is \emph{how} various cocycles are related.

Let $a^*$ and $b^*$ denote the dual basis of $\twodim^*$ and set
\[ x:=a^* \quad\text{and}\quad \xi:=-b^*. \]
Then $x$ has degree zero, $\xi$ has degree minus-one, and the definitions above become
\[ d\xi=-x, \qquad \{x,\xi\}=\langle x,\xi\rangle^{-1}=1=\langle\xi,x\rangle^{-1}=\{\xi,x\}. \]
Observe that $\NCHam{\twodim}$, the space of cyclic tensor powers, is isomorphic to $\mathbb{C}[x]$ in degree zero, as cyclic words in one generator correspond to symmetric words in one generator.

Writing $\twodim=\mathbb{C}\oplus\Sigma^{-1}\mathbb{C}$ where the even generator $a$ sits in the left-hand summand and the odd generator $b$ sits in the right-hand summand,
we have the decomposition
\[ \mat{N}{\twodim} = \twodim\otimes\mat{N}{\mathbb{C}} = \mat{N}{\mathbb{C}}\oplus\Sigma^{-1}\mat{N}{\mathbb{C}}. \]
The Hermitian matrices $\her{N}$ sit inside the left-hand summand as a subspace and hence we may restrict any polynomial superfunction $f$ in $\BV{\mat{N}{\twodim}}$ to a complex-valued polynomial function on Hermitian matrices $\her{N}$. In this way we define the \emph{expectation value}
\begin{equation} \label{eqn_expvalmap}
\langle-\rangle:\BV{\mat{N}{\twodim}}\to\mathbb{C}
\end{equation}
by
\[ \langle f\rangle:=\frac{\int_{\her{N}}f(X)e^{-\frac{1}{2}\tr(X^2)}\mathrm{d}X}{\int_{\her{N}}e^{-\frac{1}{2}\tr(X^2)}\mathrm{d}X}. \]

The following is Proposition 5.2 of \cite{GGHZ}. It follows from some standard arguments in the Batalin-Vilkovisky formalism, which encode the integration by parts relations for the Gaussian Unitary Ensemble.

\begin{prop}
The expectation value map \eqref{eqn_expvalmap} is a quasi-isomorphism of complexes whose one-sided inverse is the inclusion of $\mathbb{C}$ inside $\BV{\mat{N}{\twodim}}$ as the constant polynomials.
\end{prop}

It is helpful at this point to summarize the present situation by the following (incomplete) commutative diagram.
\begin{equation} \label{fig_expvals}
\xymatrix{
\NCBVg{\twodim} \ar@<0.5ex>@{-->}[rrr]^{\exists ?} \ar[dd]_{\sigma_{\gamma,\nu}\circ\mor_{\gamma,\nu}} \ar[rd]^(0.6){\gamma=1} & & & \mathbb{C}[\gamma,\nu] \ar@<0.5ex>[lll] \ar[ld]_(0.6){\gamma=1} \ar[dd]^{\gamma=1,\nu=N} \\
& \NCBV{\twodim} \ar@<0.5ex>@{-->}[r]^{\exists ?} \ar[ld]^(0.4){\sigma_\nu\circ\mor_\nu} & \mathbb{C}[\nu] \ar@<0.5ex>[l] \ar[rd]_(0.4){\nu=N} \\
\BV{\mat{N}{\twodim}} \ar@<0.5ex>[rrr]^{\langle - \rangle:f\mapsto\langle f\rangle} & & & \mathbb{C} \ar@<0.5ex>[lll]}
\end{equation}
The left vertical map is what we referred to earlier in Remark \ref{rem_qLQTmaps} as the quantum LQT map; in degree zero (the most important degree for us) this map sends a cyclic word $(x^k)$ to the function $\tr(X^k)$ on matrices, and it sends a symmetric product of cyclic words $(x^{k_1}) \cdots (x^{k_n})$ to the multi-trace function $\tr(X^{k_1}) \cdots \tr(X^{k_n})$. All the horizontal arrows are quasi-isomorphisms; this fact follows from a standard spectral sequence argument that is spelled out in detail in Proposition 5.2 and 5.3 of~\cite{GGHZ}.

The goal of the next section is to construct a \emph{noncommutative} analogue
\[ \langle - \rangle_{\gamma,\nu}: \NCBVg{\twodim} \to \mathbb{C}[\gamma,\nu] \]
of the expectation value map~\eqref{eqn_expvalmap}.
This map will fill in the dashed arrow atop Diagram \eqref{fig_expvals} and thus complete the above commutative diagram. That section culminates in Theorem~\ref{thm: diagram}.

Foreshadowing one of the punchlines of the paper, we claim that this dashed arrow $\langle - \rangle_{\gamma,\nu}$ will amount to taking a product of cyclic words, interpreting them as the vertices of a ribbon graph, and then returning a polynomial encoding the topology of those ribbon graphs that one can make from these vertices. Thanks to the commutativity of the diagram, we can evaluate that polynomial to compute the expected value of the multi-trace function determined by the product of cyclic words. In particular, if we use $p^{\gamma,\nu}_{k_1,\ldots,k_n} \in \mathbb{C}[\gamma,\nu]$ to denote the ``noncommutative expectation value''
\[\left\langle (x^{k_1})\cdots(x^{k_n}) \right\rangle_{\gamma,\nu},\]
then we will show that evaluating $p^{\gamma,\nu}_{k_1,\ldots,k_n}$ at $\gamma =1$ and $\nu = N$ recovers the multi-trace expectation value
\[ I_{k_1,\ldots,k_n}^N = \frac{\int_{\her{N}}\tr (X^{k_1})\cdots\tr (X^{k_n})e^{-\frac{1}{2}\tr (X^2)}\mathrm{d} X}{\int_{\her{N}}e^{-\frac{1}{2}\tr (X^2)}\mathrm{d} X}. \]
This result is the content of Proposition~\ref{prop_poly-corfun}.

\section{The ribbon graph expansion of the expectation value}
\label{sec: rib}

In this section we will introduce some instances of \emph{noncommutative} analogues of Feynman diagram expansions. The role of graphs (or Feynman diagrams), ubiquitous and familiar in the ``commutative'' setting due to Wick's Theorem, is replaced here by ribbon graphs in this noncommutative setting. The topology of the surfaces described by these ribbon graphs, determined by the genus and number of boundary components, is easily read off from the ribbon graph itself and encoded in a polynomial with two corresponding variables. These constructions will be used to complete Diagram~\eqref{fig_expvals}, which is the central result of this section, and thus relate the quantum LQT map to the topology of surfaces.

\subsection{Ribbon graphs}

We begin by recalling some basic material about ribbon graphs.

\begin{defi}
A ribbon graph $\Gamma$ consists of the following data:
\begin{itemize}
\item
A finite set $H(\Gamma)$ whose elements are called the \emph{half-edges} of $\Gamma$.
\item
A partition $E(\Gamma)$ of $H(\Gamma)$ into pairs. Elements of $E(\Gamma)$ are called \emph{edges}.
\item
A partition $V(\Gamma)$ of $H(\Gamma)$, whose elements are called \emph{vertices}.
\item
A cyclic ordering of each vertex $v\in V(\Gamma)$.
\end{itemize}
\end{defi}

This data may be equivalently described as follows. Note that a cyclic ordering of the half-edges of a vertex is the same thing as a cyclic permutation of all the half-edges at that vertex. Combining these permutations at every vertex gives us a permutation $\varsigma_\Gamma$ of all the half-edges of the ribbon graph. The cyclic decomposition of the permutation $\varsigma_\Gamma$ recovers the vertices of the ribbon graph $\Gamma$ along with their cyclic structure. Similarly, we may define a permutation $\kappa_\Gamma$ of the half-edges of $\Gamma$ satisfying:
\begin{itemize}
\item
$\kappa_\Gamma^2=\id$,
\item
$\kappa_\Gamma h \neq h$, for all $h\in H(\Gamma)$;
\end{itemize}
simply by transposing the half-edges of each edge.

The correspondence
\[ \Gamma \mapsto (H(\Gamma),\varsigma_\Gamma,\kappa_\Gamma) \]
determines an isomorphism between the category of ribbon graphs and the category whose objects are 3-tuples consisting of a set $H$, a permutation $\varsigma$ on $H$ and a permutation $\kappa$ on $H$ satisfying the two conditions just listed above (both categories have an obvious notion of isomorphism). Frequently, it is more convenient to define certain constructions in the latter category.

Given any ribbon graph $\Gamma$ we may construct a connected oriented surface $\surf_{\Gamma}$ with boundary, as shown in Figure~\ref{fig_ribbonsurface}. (Our construction will differ slightly from the one that the reader may already be familiar with, e.g. from \cite{kontairy}.)
Given a ribbon graph $\Gamma$ with $n$ vertices
\[ v_1,\ldots v_n \]
of valency $k_1,\ldots, k_n$, we take a sphere $S^2$ and remove $n$ disks, creating $k_i$ parameterized intervals around the boundary of each disk.
We then use the cyclic structure at each vertex $v_i$ to place the half-edges incident to $v_i$ into one-to-one correspondence with the parameterized intervals of the $i$th disk.
Now for every edge $e$ of $\Gamma$ we take a strip $I\times I$ and glue the ends of this strip to the two parameterized intervals corresponding to the two half-edges that form $e$.
This gives a well-defined (up to homeomorphism) construction of an oriented surface $\surf_{\Gamma}$, since we may always rotate and permute the parameterized boundary components by a homeomorphism of the sphere.
\begin{figure}[htp]
\centering
\scalebox{0.40}{\includegraphics{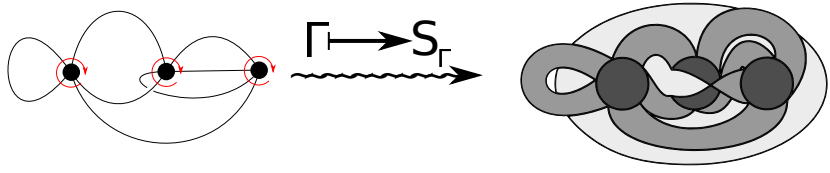}}
\caption{The surface associated to a ribbon graph. We note that the left-most central band is twisted in the picture in order to maintain the orientability of the surface.}
\label{fig_ribbonsurface}
\end{figure}
\\ The Euler characteristic of the surface $\surf_{\Gamma}$ is
\begin{equation} \label{eqn_eulerchar}
\chi(\surf_\Gamma) = 2 - |V(\Gamma)| - |E(\Gamma)|,
\end{equation}
as can be checked directly.

\begin{defi}
Given a ribbon graph $\Gamma=(H(\Gamma),\varsigma_\Gamma,\kappa_\Gamma)$, consider the permutation $\beta_\Gamma:=\varsigma_\Gamma\kappa_\Gamma$ and define the \emph{dual ribbon graph} by
\[ \Gamma^\star := (H(\Gamma),\beta_\Gamma,\kappa_\Gamma). \]
\end{defi}

The significance of the dual graph is the following. The boundary components of the surface $\surf_\Gamma$ are in one-to-one correspondence with the vertices of the dual graph $\Gamma^\star$. Using this and the formula \eqref{eqn_eulerchar} for the Euler characteristic of $\surf_\Gamma$ we get a simple expression for the genus. Define $b(\Gamma):=|V(\Gamma^\star)|$ to be the number of boundary components of $\surf_\Gamma$ and set
\begin{equation} \label{eqn_genbdryequivalence}
g(\Gamma):=\frac{|V(\Gamma)|+|E(\Gamma)|-b(\Gamma)}{2}
\end{equation}
to be the genus of the surface~$\surf_\Gamma$.

We now consider the operation of contracting an edge in a ribbon graph. This is most conveniently described in terms of the dual graph.

\begin{defi}
Let $\Gamma$ be a ribbon graph and let $e\in E(\Gamma)$ be an edge. To contract the edge $e$ in $\Gamma$, we simply remove this edge from the dual graph $\Gamma^\star$; that is if $\Gamma/e$ denotes the graph $\Gamma$ with the edge $e$ contracted then
\[ (\Gamma/e)^\star = (H(\Gamma)-e,\beta_\Gamma/e,\kappa_\Gamma/e) \]
where $\sigma/e$ denotes the permutation $\sigma$ with the entries from $e$ deleted from its cyclic decomposition.
\end{defi}

This operation may be more directly described as follows. If the edge $e=\{h,h'\}$ joins two different vertices
\[ (h h_1\ldots h_k) \quad\text{and}\quad (h' h'_1\ldots h'_{k'}) \]
of $\Gamma$ then contracting the edge $e$ combines these vertices into a single vertex
\begin{equation} \label{eqn_contractedge}
(h_1\ldots h_k h'_1\ldots h'_{k'}).
\end{equation}

If $e=\{h,h'\}$ is a loop on a vertex
\[ v=(h h_1\ldots h_k h' h'_1\ldots h'_{k'}) \]
of $\Gamma$ then contracting $e$ splits $v$ into two separate vertices
\begin{equation} \label{eqn_contractloop}
(h_1\ldots h_k) \quad\text{and}\quad (h'_1\ldots h'_{k'}).
\end{equation}

Since contracting the edge $e=\{h,h'\}$ simply deletes the half-edges from the dual graph $\Gamma^\star$, this will not change the number of vertices in $\Gamma^\star$, unless of course some of those vertices are left with no remaining half-edges. The reader may observe that this situation occurs precisely when an empty vertex is produced as a result in either \eqref{eqn_contractedge} or~\eqref{eqn_contractloop}.

If $k=k'=0$ in \eqref{eqn_contractedge} then
\begin{equation} \label{eqn_contractedgebdry}
b(\Gamma/e) = b(\Gamma)-1.
\end{equation}

If either $k=0$ or $k'=0$ (but not both) in \eqref{eqn_contractloop} then
\begin{equation} \label{eqn_contractlooponebdry}
b(\Gamma/e) = b(\Gamma)-1.
\end{equation}
and if both $k=k'=0$ then
\begin{equation} \label{eqn_contractlooptwobdry}
b(\Gamma/e) = b(\Gamma)-2.
\end{equation}

The effect on the genus of contracting an edge may be summarized as follows. Contracting an edge $e$ that joins two different vertices, as in \eqref{eqn_contractedge}, typically decreases the number of vertices and edges by one each and hence in this case
\begin{equation} \label{eqn_contractedgegenus}
g(\Gamma/e) = g(\Gamma)-1.
\end{equation}

If the edge $e$ is instead a loop then similar reasoning shows that the genus is not affected
\begin{equation} \label{eqn_contractloopgenus}
g(\Gamma/e) = g(\Gamma).
\end{equation}
Note that a careful analysis shows that \eqref{eqn_contractedgegenus} and \eqref{eqn_contractloopgenus} still hold even if the number $b(\Gamma)$ is affected by the edge contraction as above, as in this case the other terms in \eqref{eqn_genbdryequivalence} will compensate appropriately.

\subsection{Chord diagrams}

We will use chord diagrams in our construction of the maps that fill in Diagram~\eqref{fig_expvals} and later in Section~\ref{sec_largeN}.
We begin by recalling their definition.

\begin{defi}
Given a set $X$ of even cardinality, a chord diagram on $X$ is a partition of $X$ into pairs. We denote the set of all chord diagrams on $X$ by $\chord{X}$. In particular, if $X$ is the set of all integers between $1$ and $2l$ then we denote this set by $\chord{2l}$.
\end{defi}

If $V$ is a graded vector space equipped with a symmetric bilinear form $B$ of even degree,
then for every chord diagram
\[ c=\{i_1,j_1\},\{i_2,j_2\},\ldots,\{i_l,j_l\}\in\chord{2l}, \]
we may define a map $B_c:V^{\otimes 2l}\to\gf$ by
\[ B_c(v_1,v_2,\ldots,v_{2l}) = \pm B(v_{i_1},v_{j_1}) B(v_{i_2},v_{j_2}) \cdots B(v_{i_l},v_{j_l}) \]
where the sign is determined canonically by the Koszul sign rule. Note that
\begin{equation} \label{eqn_chordpermform}
B_c(\sigma\cdot x) = B_{\sigma^{-1}\cdot c}(x)
\end{equation}
for all $x\in V^{\otimes 2l}$, $c\in\chord{2l}$, and $\sigma\in\symg{2l}$.

We may also use chord diagrams to define ribbon graphs as follows.

\begin{defi}
Given a chord diagram
\[ c=\{i_1,j_1\},\{i_2,j_2\},\ldots,\{i_l,j_l\}\in\chord{2l} \]
and a list of positive integers $k_1,\ldots k_n$ such that
\[ k_1+k_2+\cdots+k_n=2l, \]
let $\Gamma_{k_1,\ldots,k_n}(c)$ denote the ribbon graph  with $2l$ half edges $h_1,\ldots,h_{2l}$. The edges of $\Gamma_{k_1,\ldots,k_n}(c)$ are
\[ \{h_{i_1},h_{j_1}\}, \{h_{i_2},h_{j_2}\}, \ldots, \{h_{i_l},h_{j_l}\} \]
and the vertices are defined by the permutation
\[ (h_1\ldots h_{k_1})(h_{k_1+1}\ldots h_{k_1+k_2})\cdots(h_{k_1+\cdots+k_{n-1}+1}\ldots h_{k_1+\cdots+k_{n-1}+k_n}).\]
\end{defi}

\subsection{Expectation values in noncommutative geometry}
\label{sec: exp in NC}

We are now ready to proceed with our construction of the maps that complete Diagram~\eqref{fig_expvals}. As a gloss -- and because our constructions should generalize to any contractible space, though there is no present need to work at such a level of generality -- we will provide a ribbon graph version of the computation of Gaussian moments by pairing off legs of a vertex. Hence we begin by describing our ``propagator'', then describe our Wick formula, and finally verify that this construction fits into Diagram \eqref{fig_expvals}.

On our two-dimensional symplectic vector space $\twodim$, we form a degree zero symmetric pairing
\[ B:\twodim\otimes\twodim\to\mathbb{C}\]
where
\[ B(u,v) := \langle du,v \rangle \]
using the symplectic form $\innprod$ and differential $d$.
This pairing is degenerate, but becomes nondegenerate on the subspace of $\twodim$ concentrated in degree zero. Let
\[ B^{-1} = a\otimes a \in\twodim\otimes\twodim \]
denote the inverse of this form restricted to this subspace.

Now, given a list of positive integers $k_1,\ldots,k_n$, we define a map
\[ \langle-\rangle_{\gamma,\nu}^{k_1,\ldots, k_n}:(\twodim^*)^{\otimes k_1}\otimes\cdots\otimes(\twodim^*)^{\otimes k_n}\to\mathbb{C}[\gamma,\nu]. \]
The map is zero if the sum of the $k_i$ is odd.
If the sum is even and equal to $2l$, then we define
\begin{equation} \label{eqn_expvalmapNCgk}
\langle w\rangle_{\gamma,\nu}^{k_1,\ldots,k_n}:=\sum_{c\in\chord{2l}} B^{-1}_c(w)\gamma^{g(\Gamma_{k_1,\ldots,k_n}(c))}\nu^{b(\Gamma_{k_1,\ldots,k_n}(c))}
\end{equation}
where $w\in(\twodim^*)^{\otimes k_1}\otimes\cdots\otimes(\twodim^*)^{\otimes k_n}=(\twodim^*)^{\otimes 2l}$.

This map plays a central role in what follows, so we make a few orienting remarks. Note that it records the topology of those ribbon graphs that are built from chord diagrams using the prescribed data consisting of the list of the $k_i$s. The input $w$ is a kind of ``noncommutative function'' on $\twodim$ (i.e., an element of the tensor algebra of the dual). Those familiar with Feynman diagrams will see the resemblance to the usual Wick formula.

It follows from \eqref{eqn_chordpermform} that this map is cyclically symmetric in the variables. A permutation $\sigma\in\cycg{k_1}\times\ldots\times\cycg{k_n}$
determines an isomorphism
\[ \sigma:\Gamma_{k_1,\ldots,k_n}(c)\cong\Gamma_{k_1,\ldots,k_n}(\sigma\cdot c). \]
Hence $\langle \sigma\cdot w\rangle_{\gamma,\nu}^{k_1,\ldots,k_n} = \langle w\rangle_{\gamma,\nu}^{k_1,\ldots,k_n}$ and \eqref{eqn_expvalmapNCgk} extends to give a well-defined map on $\NCHamP{\twodim}^{\otimes n}$.
Similar reasoning shows that, by permuting the vertices of the ribbon graph, this map extends to give a well-defined map on $\sym{\NCHamP{\twodim}}$.
(Note that on $S^0\!\left(\NCHamP{\twodim}\right)=\mathbb{C}$, this map is just the inclusion of $\mathbb{C}$ into $\mathbb{C}[\gamma,\nu]$.)
If we extend this map linearly with respect to the variables $\gamma$ and $\nu$,
we see that \eqref{eqn_expvalmapNCgk} provides a well-defined map
\begin{equation} \label{eqn_expvalmapNCg}
\langle - \rangle_{\gamma,\nu}:\NCBVg{\twodim}\to\mathbb{C}[\gamma,\nu].
\end{equation}
We call this map the \emph{noncommutative expectation value}. This map will provide the desired dashed arrow that completes Diagram~\eqref{fig_expvals}, but first we must prove some of its basic properties.

\begin{theorem} \label{thm_ncexpval}
The noncommutative expectation value map \eqref{eqn_expvalmapNCg} is a quasi-isomorphism of complexes whose one-sided inverse is the inclusion of $\mathbb{C}[\gamma,\nu]$ into~$\NCBVg{\twodim}$.
\end{theorem}

\begin{proof}
We must show first that \eqref{eqn_expvalmapNCg} is a map of complexes; that is, we must show that it vanishes on the boundaries of $\NCBVg{\twodim}$. To see that the map is a quasi-isomorphism, first note that the cohomology of $\NCBVg{\twodim}$ is $\mathbb{C}[\gamma,\nu]$ because there is a filtration of $\NCBVg{\twodim}$ by powers of $\twodim$ and the associated spectral sequence collapses immediately as the complex $\twodim$ is acyclic. Now apply the comparison theorem.

Recall that $\twodim$ is $\mathbb{Z}$-graded and concentrated in positive degrees, so that $\twodim^*$ and hence $\NCBVg{\twodim}$ are also $\mathbb{Z}$-graded and concentrated in nonpositive degrees.
To show \eqref{eqn_expvalmapNCg} is a cochain map,
it thus suffices only to consider the boundaries of degree minus-one elements in $\NCBVg{\twodim}$.
Such elements are linear combinations of terms of the form
\[ (\xi x^i)(x^{k_1})\cdots(x^{k_r}) \]
where $i,r\geq 0$ and $k_1,\ldots,k_r\geq 1$.

The boundary of such a typical element is
\begin{multline*}
(d+\Delta_{\gamma,\nu})\left[(\xi x^i)(x^{k_1})\cdots(x^{k_r})\right] = -(x^{i+1})(x^{k_1})\cdots(x^{k_r}) + \sum_{s=1}^i (x^{s-1})(x^{i-s})(x^{k_1})\cdots(x^{k_r}) \\
+ \ \gamma\sum_{t=1}^r k_t(x^{i+k_t-1})(x^{k_1})\cdots\widehat{(x^{k_t})}\cdots(x^{k_r}).
\end{multline*}
Hence we must show that
\begin{multline}
\label{eqn_cocycleidentity}
\left\langle (x^{i+1})(x^{k_1})\cdots(x^{k_r}) \right\rangle_{\gamma,\nu} = \sum_{s=1}^i \left\langle (x^{s-1})(x^{i-s})(x^{k_1})\cdots(x^{k_r}) \right\rangle_{\gamma,\nu} \\
+ \ \gamma\sum_{t=1}^r k_t\left\langle (x^{i+k_t-1})(x^{k_1})\cdots\widehat{(x^{k_t})}\cdots(x^{k_r}) \right\rangle_{\gamma,\nu}.
\end{multline}
Obviously we may assume that $i+k_1+\cdots k_r = 2l-1$ for some $l\geq 1$.
We begin by calculating the left-hand side of \eqref{eqn_cocycleidentity}.
As a point of notation, note that any chord diagram $c\in\chord{2l}$ defines in an obvious way a permutation in $\symg{2l}$, also denoted by $c$, satisfying $c^2=\id$.
Now
\begin{displaymath}
\begin{split}
\left\langle (x^{i+1})(x^{k_1})\cdots(x^{k_r}) \right\rangle_{\gamma,\nu} =& \sum_{c\in\chord{2l}} \gamma^{g\left(\Gamma_{i+1,k_1,\ldots,k_r}(c)\right)}\nu^{b\left(\Gamma_{i+1,k_1,\ldots,k_r}(c)\right)} \\
=& \sum_{s=2}^{2l}\sum_{\begin{subarray}{c} \{ c\in\chord{2l} \,:\, c(1)=s \}\end{subarray}} \gamma^{g\left(\Gamma_{i+1,k_1,\ldots,k_r}(c)\right)}\nu^{b\left(\Gamma_{i+1,k_1,\ldots,k_r}(c)\right)} \\
=& \sum_{s=1}^i\sum_{\begin{subarray}{c}\{ c\in\chord{2l} \,:\, c(1)=s+1 \}\end{subarray}} \gamma^{g\left(\Gamma_{i+1,k_1,\ldots,k_r}(c)\right)}\nu^{b\left(\Gamma_{i+1,k_1,\ldots,k_r}(c)\right)} \\
&+ \sum_{t=1}^r\sum_{u=1}^{k_t}\sum_{\left\{\begin{subarray}{c} c\in\chord{2l}: \\ c(1)=i+1+k_1+\cdots+k_{t-1}+u \end{subarray}\right\}}\gamma^{g\left(\Gamma_{i+1,k_1,\ldots,k_r}(c)\right)}\nu^{b\left(\Gamma_{i+1,k_1,\ldots,k_r}(c)\right)}.
\end{split}
\end{displaymath}

Hence to prove \eqref{eqn_cocycleidentity}, we will show that
\begin{equation} \label{eqn_cocycleidentitycobracket}
\left\langle (x^{s-1})(x^{i-s})(x^{k_1})\cdots(x^{k_r}) \right\rangle_{\gamma,\nu} = \sum_{\begin{subarray}{c} \{ c\in\chord{2l} \,:\,  c(1)=s+1\} \end{subarray}} \gamma^{g\left(\Gamma_{i+1,k_1,\ldots,k_r}(c)\right)}\nu^{b\left(\Gamma_{i+1,k_1,\ldots,k_r}(c)\right)}
\end{equation}
for $s=1,\ldots,i$ and that
\begin{multline} \label{eqn_cocycleidentitybracket}
\gamma\left\langle (x^{i+k_t-1})(x^{k_1})\cdots\widehat{(x^{k_t})}\cdots(x^{k_r}) \right\rangle_{\gamma,\nu} = \\
\sum_{\left\{\begin{subarray}{c} c\in\chord{2l}: \\ c(1)=i+1+k_1+\cdots+k_{t-1}+s \end{subarray}\right\}}\gamma^{g\left(\Gamma_{i+1,k_1,\ldots,k_r}(c)\right)}\nu^{b\left(\Gamma_{i+1,k_1,\ldots,k_r}(c)\right)}
\end{multline}
for $t=1,\ldots r$ and $s=1,\ldots, k_t$.

We prove \eqref{eqn_cocycleidentitycobracket} first.
Note that when $i=0$ there is nothing to prove,
so we begin by considering the case $i=1=s$.
In this case the left-hand side of \eqref{eqn_cocycleidentitycobracket} is $\nu^2\left\langle (x^{k_1})\cdots(x^{k_r}) \right\rangle_{\gamma,\nu}$
and the right-hand side may be written as
\[ \sum_{c\in\chord{\{3,\ldots,2l\}}}\gamma^{g\left(\Gamma_{2,k_1,\ldots,k_r}(c\cup\{\{1,2\}\})\right)}\nu^{b\left(\Gamma_{2,k_1,\ldots,k_r}(c\cup\{\{1,2\}\})\right)}. \]
Contracting the edge $e=\{h_1,h_2\}$ in the above and applying \eqref{eqn_contractlooptwobdry} and \eqref{eqn_contractloopgenus},
we see that this may be written as
\[ \nu^2\sum_{c\in\chord{2l-2}}\gamma^{g\left(\Gamma_{k_1,\ldots,k_r}(c)\right)}\nu^{b\left(\Gamma_{k_1,\ldots,k_r}(c)\right)} = \nu^2\left\langle (x^{k_1})\cdots(x^{k_r}) \right\rangle_{\gamma,\nu}, \]
as desired.

Next we consider the case when $i\geq 2$ and $s=1$ or $s=i$. For simplicity, assume $s=1$. In this case the left-hand side of \eqref{eqn_cocycleidentitycobracket} is $\nu\left\langle(x^{i-1})(x^{k_1})\cdots(x^{k_r}) \right\rangle_{\gamma,\nu}$ and the right-hand side may be written as
\[ \sum_{c\in\chord{\{3,\ldots,2l\}}} \gamma^{g\left(\Gamma_{i+1,k_1,\ldots,k_r}(c\cup\{\{1,2\}\})\right)} \nu^{b\left(\Gamma_{i+1,k_1,\ldots,k_r}(c\cup\{\{1,2\}\})\right)}. \]
Contracting the edge $e=\{h_1,h_2\}$ and applying \eqref{eqn_contractlooponebdry} and \eqref{eqn_contractloopgenus},
this becomes
\[ \nu\sum_{c\in\chord{2l-2}} \gamma^{g\left(\Gamma_{i-1,k_1,\ldots,k_r}(c)\right)} \nu^{b\left(\Gamma_{i-1,k_1,\ldots,k_r}(c)\right)} = \nu\left\langle(x^{i-1})(x^{k_1})\cdots(x^{k_r}) \right\rangle_{\gamma,\nu}. \]

Finally, consider the case when $1<s<i$.
In this case we may write the right-hand side of \eqref{eqn_cocycleidentitycobracket} as
\[ \sum_{c\in\chord{\{1,\ldots,2l\}-\{1,s+1\}}} \gamma^{g\left(\Gamma_{i+1,k_1,\ldots,k_r}(c\cup\{\{1,s+1\}\})\right)}\nu^{b\left(\Gamma_{i+1,k_1,\ldots,k_r}(c\cup\{\{1,s+1\}\})\right)}.\]
Contracting the edge $e=\{h_1,h_{s+1}\}$, this becomes
\[ \sum_{c\in\chord{2l-2}} \gamma^{g\left(\Gamma_{s-1,i-s,k_1,\ldots,k_r}(c)\right)} \nu^{b\left(\Gamma_{s-1,i-s,k_1,\ldots,k_r}(c)\right)} = \left\langle (x^{s-1})(x^{i-s})(x^{k_1})\cdots(x^{k_r}) \right\rangle_{\gamma,\nu}. \]

It remains to prove \eqref{eqn_cocycleidentitybracket}. We first consider the case $i=0$ and $k_t=1=s$. In this case the left-hand side of \eqref{eqn_cocycleidentitybracket} is $\gamma\nu\left\langle (x^{k_1})\cdots\widehat{(x^{k_t})}\cdots(x^{k_r}) \right\rangle_{\gamma,\nu}$. Proceeding as above, but instead making use of \eqref{eqn_contractedgebdry} and \eqref{eqn_contractedgegenus}, we may write the right-hand side of \eqref{eqn_cocycleidentitybracket} as
\[ \gamma\nu\sum_{c\in\chord{2l-2}} \gamma^{g\left(\Gamma_{k_1,\ldots,\widehat{k_t},\ldots,k_r}(c)\right)} \nu^{b\left(\Gamma_{k_1,\ldots,\widehat{k_t},\ldots,k_r}(c)\right)}, \]
as required.

Finally, we consider the case $i+k_t\geq 2$. Repeating the same arguments as above, but making use of \eqref{eqn_contractedgegenus}, we may write the right-hand side of \eqref{eqn_cocycleidentitybracket} as
\begin{multline*}
\gamma\sum_{c\in\chord{2l-2}} \gamma^{g\left(\Gamma_{i+k_t-1,k_1,\ldots,\widehat{k_t},\ldots,k_r}(c)\right)} \nu^{b\left(\Gamma_{i+k_t-1,k_1,\ldots,\widehat{k_t},\ldots,k_r}(c)\right)} = \\
\gamma\left\langle (x^{i+k_t-1})(x^{k_1})\cdots\widehat{(x^{k_t})}\cdots(x^{k_r}) \right\rangle_{\gamma,\nu},
\end{multline*}
as desired.
\end{proof}

Since $\langle - \rangle_{\gamma,\nu}$ is $\gamma$-linear,
it determines a unique map
\begin{equation} \label{eqn_expvalmapNC}
\langle - \rangle_\nu:\NCBV{\twodim}\to\mathbb{C}[\nu]
\end{equation}
that makes the upper quadrilateral of Diagram~\eqref{fig_expvals} commute by setting $\gamma = 1$.
This map may be described explicitly as the $\mathbb{C}[\nu]$-linear map that is the identity on $S^0(\NCHamP{\twodim})=\mathbb{C}$ and such that
\begin{equation} \label{eqn_expvalmapNCk}
\langle w\rangle_\nu^{k_1,\ldots,k_n} = \sum_{c\in\chord{2l}} B^{-1}_c(w) \nu^{b(\Gamma_{k_1,\ldots,k_n}(c))}
\end{equation}
where $w\in(\twodim^*)^{\otimes k_1}\otimes\cdots\otimes(\twodim^*)^{\otimes k_n}=(\twodim^*)^{\otimes 2l}$.
This observation leads us to the following corollary of Theorem~\ref{thm_ncexpval}.

\begin{cor} \label{cor_ncexpval}
The map \eqref{eqn_expvalmapNC} is a quasi-isomorphism of complexes whose one-sided inverse is the inclusion of $\mathbb{C}[\nu]$ into $\NCBV{\twodim}$.
\end{cor}

\begin{proof}
It follows tautologically from Theorem \ref{thm_ncexpval} that \eqref{eqn_expvalmapNC} is a map of complexes and the same spectral sequence argument used there applies to show that it is a quasi-isomorphism.
\end{proof}

Finally, it remains to show that these maps fit into Diagram~\eqref{fig_expvals}.

\begin{theorem}
\label{thm: diagram}
Placing the noncommutative expectation value maps \eqref{eqn_expvalmapNCg} and \eqref{eqn_expvalmapNC} in Diagram~\eqref{fig_expvals} produces a commutative diagram in which all the horizontal arrows are quasi-isomorphisms.
\end{theorem}

\begin{proof}
All that remains is to show that the resulting diagram is commutative. For this, it is sufficient to prove that the outside rectangle commutes. Since all the maps in this rectangle have degree zero it follows that in nonzero degrees the horizontal maps send everything to zero, and so in this case the diagram commutes trivially. Therefore the only interesting case is in degree zero because all four corners are nontrivial in this degree.

Note that we have already established that each row of this diagram describes a retract of the left-hand complex onto its cohomology: for $\NCBVg{\twodim}$ the arrows retract it onto $\mathbb{C}[\gamma,\nu]$. Now, any degree zero element $f$ in $\NCBVg{\twodim}$ is a cocycle, since $\NCBVg{\twodim}$ is nonpositively graded and the differential has degree one. Using the retraction, we know that there is some polynomial $p_f \in \mathbb{C}[\gamma,\nu]$ cohomologous to $f$. Hence $f - p_f$ is a boundary in $\NCBVg{\twodim}$, and so it is sent to zero along either route around the diagram, as all the maps in Diagram~\eqref{fig_expvals} are chain maps. As a consequence, when we follow the path around the northeast periphery of the diagram, both $f$ and $p_f$ will go to the same number in $\mathbb{C}$. The same will also be true when we take the other route. Now note that $p_f$ is sent to the same number in $\mathbb{C}$ along either of these two routes; here we have used the fact that the horizontal arrows are retractions. It now follows that the same is true of~$f$.
\end{proof}

\section{Correlation functions in the Gaussian Unitary Ensemble}
\label{sec: corr}

In this section we apply our preceding results to describe the correlation functions \eqref{eqn_corfunmultitrace} of the Gaussian Unitary Ensemble.

\subsection{The correlation functions and a recursive formula}

\begin{defi}
Consider the monomial
\begin{equation} \label{eqn_principleclass}
(x^{k_1})(x^{k_2})\cdots(x^{k_n})\in\left[\NCHam{\twodim}^{\otimes n}\right]_{\symg{n}}
\end{equation}
given by a product of cyclic words on one generator.
Define the polynomial $p^{\gamma,\nu}_{k_1,\ldots,k_n}\in\mathbb{C}[\gamma,\nu]$ by
\[ p^{\gamma,\nu}_{k_1,\ldots,k_n} := \left\langle (x^{k_1})\cdots(x^{k_n}) \right\rangle_{\gamma,\nu},\]
i.e., as the noncommutative expectation value of this monomial.
Likewise, define the polynomial $p^{\nu}_{k_1,\ldots,k_n}\in\mathbb{C}[\nu]$ by
\[ p^{\nu}_{k_1,\ldots,k_n} := \left\langle (x^{k_1})\cdots(x^{k_n}) \right\rangle_\nu. \]
\end{defi}

By Theorem \ref{thm_ncexpval} the polynomial $p^{\gamma,\nu}_{k_1,\ldots,k_n}$ is the unique representative in $\mathbb{C}[\gamma,\nu]$ of the cohomology class~\eqref{eqn_principleclass} lying in $\NCBVg{\twodim}$.
Likewise, by Corollary \ref{cor_ncexpval}, $p^\nu_{k_1,\ldots,k_n}$ is the unique representative in $\mathbb{C}[\nu]$ of the cohomology class~\eqref{eqn_principleclass} lying in $\NCBV{\twodim}$. Since Diagram~\eqref{fig_expvals} commutes, the polynomial $p^\nu_{k_1,\ldots,k_n}$ is obtained from $p^{\gamma,\nu}_{k_1,\ldots,k_n}$ by setting~$\gamma=1$.

In what follows we will chiefly work with the polynomials $p^\nu_{k_1,\ldots,k_n}$, since it will become apparent that the parameter $\gamma$ has no role to play in analyzing the correlation functions of the Gaussian Unitary Ensemble; however, it makes sense to include the polynomials $p^{\gamma,\nu}_{k_1,\ldots,k_n}$ in our discussion since they may be studied and computed using the same techniques that will be explained here -- tracking the genus can be useful in closely related settings.

Note that the polynomials $p^{\gamma,\nu}_{k_1,\ldots,k_n}$ and $p^\nu_{k_1,\ldots,k_n}$ are well-defined for nonnegative values of the indices $k_r$ with
\[ p^{\gamma,\nu}_{0,k_1,\ldots,k_n} = \nu p^{\gamma,\nu}_{k_1,\ldots,k_n} \quad\text{and}\quad p^\nu_{0,k_1,\ldots,k_n} = \nu p^\nu_{k_1,\ldots,k_n}. \]
This is because $x^0$ is precisely~$\nu$.

From the formulas \eqref{eqn_expvalmapNCgk} and \eqref{eqn_expvalmapNCk} for $\langle - \rangle_{\gamma,\nu}$ and $\langle - \rangle_\nu$,
we know that the coefficients of the polynomials $p^{\gamma,\nu}_{k_1,\ldots,k_n}$ and $p^\nu_{k_1,\ldots,k_n}$ may be described by counting chord diagrams.
Given positive integers $k_1,\ldots, k_n$ such that $k_1+\cdots+k_n = 2l$, set
\begin{displaymath}
\begin{split}
\eta^{g,b}_{k_1,\ldots, k_n} &:= \left|\left\{c\in\chord{2l}:g(\Gamma_{k_1,\ldots,k_n}(c)) = g \text{ and } b(\Gamma_{k_1,\ldots,k_n}(c)) = b\right\}\right|, \\
\eta^{g,*}_{k_1,\ldots, k_n} &:= \left|\left\{c\in\chord{2l}:g(\Gamma_{k_1,\ldots,k_n}(c)) = g\right\}\right|, \\
\eta^{*,b}_{k_1,\ldots, k_n} &:= \left|\left\{c\in\chord{2l}: b(\Gamma_{k_1,\ldots,k_n}(c)) = b\right\}\right|.
\end{split}
\end{displaymath}
Then we may write
\begin{equation} \label{eqn_polycoeffs}
\begin{split}
p^{\gamma,\nu}_{k_1,\ldots,k_n} &= \sum_{g\geq 0, b\geq 1} \eta^{g,b}_{k_1,\ldots, k_n}\gamma^g\nu^b, \\
p^\nu_{k_1,\ldots,k_n} &= \sum_{b\geq 1} \eta^{*,b}_{k_1,\ldots, k_n}\nu^b.
\end{split}
\end{equation}
Recall from \eqref{eqn_genbdryequivalence} that for $k_1,\ldots,k_n$ as above
\[ 2g\left(\Gamma_{k_1,\ldots,k_n}(c)\right) +b\left(\Gamma_{k_1,\ldots,k_n}(c)\right) = n + l. \]
From this, it follows that when  $2g+b = n+l$, we have
\[ \eta^{g,b}_{k_1,\ldots, k_n} = \eta^{g,*}_{k_1,\ldots, k_n} = \eta^{*,b}_{k_1,\ldots, k_n} \]
and zero otherwise. Hence we may write
\begin{displaymath}
\begin{split}
p^{\gamma,\nu}_{k_1,\ldots,k_n} &= \sum_{\begin{subarray}{c} g\geq 0, b\geq 1: \\ 2g+b=n+l \end{subarray}} \eta^{g,b}_{k_1,\ldots, k_n}\gamma^g\nu^b, \\ &= \sum_{\begin{subarray}{c} g\geq 0, b\geq 1: \\ 2g+b=n+l \end{subarray}} \eta^{g,*}_{k_1,\ldots, k_n}\gamma^g\nu^b = \sum_{\begin{subarray}{c} g\geq 0, b\geq 1: \\ 2g+b=n+l \end{subarray}} \eta^{*,b}_{k_1,\ldots, k_n}\gamma^g\nu^b, \\
&= \sum_{g=0}^{\left[\frac{n+l-1}{2}\right]}\eta^{g,*}_{k_1,\ldots, k_n}\gamma^g\nu^{n+l-2g} = \sum_{g=0}^{\left[\frac{n+l-1}{2}\right]}\eta^{*,n+l-2g}_{k_1,\ldots, k_n}\gamma^g\nu^{n+l-2g}
\end{split}
\end{displaymath}
and
\[p^\nu_{k_1,\ldots,k_n} = \sum_{g=0}^{\left[\frac{n+l-1}{2}\right]}\eta^{g,*}_{k_1,\ldots, k_n}\nu^{n+l-2g} = \sum_{g=0}^{\left[\frac{n+l-1}{2}\right]}\eta^{*,n+l-2g}_{k_1,\ldots, k_n}\nu^{n+l-2g}.\]
From this expression it is clear that the polynomials $p^\nu_{k_1,\ldots,k_n}$ are even (respectively, odd) precisely when $(n+l)$ is even (respectively, odd).

The essential significance of the polynomials $p^\nu_{k_1,\ldots,k_n}$ -- of principal importance for this paper -- is that they recover all the multi-trace correlation functions~\eqref{eqn_corfunmultitrace},
thanks to Proposition 5.4 of \cite{GGHZ}, which we now recall.

\begin{prop}
\label{prop_poly-corfun}
For every positive integer $N$,
\[ p^\nu_{k_1,\ldots,k_n}(N) = I_{k_1,\ldots,k_n}^N = \frac{\int_{\her{N}}\tr (X^{k_1})\cdots\tr (X^{k_n})e^{-\frac{1}{2}\tr (X^2)}\mathrm{d} X}{\int_{\her{N}}e^{-\frac{1}{2}\tr (X^2)}\mathrm{d} X}. \]
\end{prop}

\begin{proof}
Although a proof is provided in \textit{loc. cit.}, we mention for the sake of completeness that the result follows from the fact that Diagram~\eqref{fig_expvals} commutes and that both $\sigma_\nu$ and $\mor_\nu$ are maps of Batalin-Vilkovisky algebras, provided we recall that the image of $(x^k)\in\NCHam{\twodim}$ under $\sigma_\nu\circ\mor_\nu$ is the polynomial function $\tr(X^k)$ for $X\in\gl{N}{\mathbb{C}}$.
\end{proof}

Our methods provide us with a recursive formula for computing these polynomials.

\begin{prop}
For all $n\geq 1$ and for all $k_1\geq 1$ and $k_2,\ldots,k_n\geq 0$,
\begin{equation} \label{eqn_recursiongenus}
p^{\gamma,\nu}_{k_1,\ldots,k_n} = \sum_{r=1}^{k_1-1} p^{\gamma,\nu}_{r-1,k_1-r-1,k_2,\ldots,k_n} + \gamma\sum_{r=2}^n k_r p^{\gamma,\nu}_{k_1+k_r-2,k_2,\ldots,\widehat{k_r},\ldots,k_n}.
\end{equation}
In particular,
\begin{equation} \label{eqn_recursionbdry}
p^\nu_{k_1,\ldots,k_n} = \sum_{r=1}^{k_1-1} p^\nu_{r-1,k_1-r-1,k_2,\ldots,k_n} + \sum_{r=2}^n k_r p^\nu_{k_1+k_r-2,k_2,\ldots,\widehat{k_r},\ldots,k_n}.
\end{equation}
\end{prop}

\begin{proof}
Observe that \eqref{eqn_recursiongenus} follows immediately from Equation \eqref{eqn_cocycleidentity} and that \eqref{eqn_recursionbdry} is obtained from \eqref{eqn_recursiongenus} simply by setting $\gamma=1$.
\end{proof}

\subsection{Examples}

Here we will provide some examples of how Equation \eqref{eqn_recursionbdry} can be used to compute the polynomials $p^\nu_{k_1,\ldots,k_n}$ and hence, by Proposition \ref{prop_poly-corfun}, the $n$-trace correlation functions.

\begin{example}
As an example of a very simple computation, we compute
\begin{displaymath}
\begin{split}
p^\nu_2 &= p^\nu_{0,0} = \nu^2, \\
p^\nu_{1,3} &= 3p^\nu_{2} = 3\nu^2, \\
p^\nu_{2,2} &= p^\nu_{0,0,2} + 2p^\nu_2 = (\nu^2+2)p^\nu_2 = (\nu^2+2)\nu^2.
\end{split}
\end{displaymath}
\end{example}

In principle, we may compute any polynomial $p^\nu_{k_1,\ldots,k_n}$ through the repeated application of the recursion relation \eqref{eqn_recursionbdry}.
As a practical matter, however, the number of terms increases rapidly at each stage.
Nonetheless, in some simple cases we obtain some general formulas.

\begin{example} \label{exm_allones}
This example was presented in \cite{GGHZ}, but since we will require it later we present it here now as well. By the recursion relation \eqref{eqn_recursiongenus} we have
\[ p^{\gamma,\nu}_{\underbrace{\scriptstyle{1,1,\ldots,1}}_{2n \text{~terms}}} = (2n-1)\gamma\cdot p^{\gamma,\nu}_{0,\underbrace{\scriptstyle{1,1,\ldots\ldots\ldots,1}}_{2n-2 \text{~terms}}} = (2n-1)\gamma\nu p^{\gamma,\nu}_{\underbrace{\scriptstyle{1,1,\ldots\ldots,1}}_{2n-2 \text{~terms}}}. \]
Therefore,
\[ p^{\gamma,\nu}_{\underbrace{\scriptstyle{1,1,\ldots,1}}_{2n \text{~terms}}} = (2n-1)!!\gamma^n\nu^n \qquad\text{and}\qquad p^\nu_{\underbrace{\scriptstyle{1,1,\ldots,1}}_{2n \text{~terms}}} = (2n-1)!!\nu^n \]
and, by Proposition \ref{prop_poly-corfun}, we have
\[ \frac{\int_{\her{N}}\tr(X)^{2n}e^{-\frac{1}{2}\tr (X^2)}\mathrm{d} X}{\int_{\her{N}}e^{-\frac{1}{2}\tr (X^2)}\mathrm{d} X} = (2n-1)!! N^n. \]
\end{example}

\begin{example}
By the recursion relation \eqref{eqn_recursiongenus} we have
\[ p^{\gamma,\nu}_{\underbrace{\scriptstyle{2,2,\ldots,2}}_{n \text{~terms}}} = p^{\gamma,\nu}_{0,0,\underbrace{\scriptstyle{2,2,\ldots\ldots,2}}_{n-1 \text{~terms}}} + 2(n-1)\gamma p^{\gamma,\nu}_{\underbrace{\scriptstyle{2,2,\ldots\ldots,2}}_{n-1 \text{~terms}}} = (\nu^2+2(n-1)\gamma) p^{\gamma,\nu}_{\underbrace{\scriptstyle{2,2,\ldots\ldots,2}}_{n-1 \text{~terms}}}. \]
Therefore,
\[ p^{\gamma,\nu}_{\underbrace{\scriptstyle{2,2,\ldots,2}}_{n \text{~terms}}} = \prod_{j=0}^{n-1}(\nu^2+2j\gamma) \qquad\text{and}\qquad p^\nu_{\underbrace{\scriptstyle{2,2,\ldots,2}}_{n \text{~terms}}} = \prod_{j=0}^{n-1}(\nu^2+2j). \]
More generally, one can show that
\[ p^{\gamma,\nu}_{\underbrace{\scriptstyle{2,2,\ldots,2}}_{n \text{~terms}},k_1,\ldots,k_m} = p^{\gamma,\nu}_{k_1,\ldots,k_m}\prod_{j=0}^{n-1}\left(\nu^2+\left(2j+\sum_{i=1}^m k_i\right)\gamma\right). \]
\end{example}

\subsection{Order of the correlation functions}

Finally, we will demonstrate how equation \eqref{eqn_recursionbdry} may be used to compute the degree of the polynomial~$p^\nu_{k_1,\ldots,k_n}$.

\begin{theorem} \label{thm_degree}
Consider the polynomial $p^\nu_{k_1,\ldots,k_n}$ where $k_1,\ldots,k_n\geq 0$ have even sum. Let
\[ m:=\frac{1}{2}\sum_{i=1}^n k_i, \]
and let \[ q:=|\{i: k_i \text{ is even}\}| \]
count the number of indices that are even. Then $p^\nu_{k_1,\ldots,k_n}$ is a polynomial of degree~$(m+q)$.
\end{theorem}

\begin{proof}
The proof proceeds by induction on $m$. We will use the fact that the coefficients of the polynomials~\eqref{eqn_polycoeffs} are nonnegative.
Note that if $m=0$, then all the indices $k_i$ are zero and the statement is trivial. If $m>0$ then we consider two cases:
\begin{itemize}
\item {\it Case 1:} there is an index which is odd.

There must then be at least two indices which are odd, and we may assume that it is both $k_1$ and $k_2$ that are odd.
From \eqref{eqn_recursionbdry} we have
\begin{displaymath}
\begin{split}
p^\nu_{k_1,\ldots,k_n} =& \sum_{r=1}^{k_1-1} p^\nu_{r-1,k_1-r-1,k_2,\ldots,k_n} \\
&+ k_2 p^\nu_{k_1+k_2-2,k_3,\ldots,k_n} + \sum_{r=3}^n k_r p^\nu_{k_1+k_r-2,k_2,\ldots,\widehat{k_r},\ldots,k_n}.
\end{split}
\end{displaymath}

Applying the inductive hypothesis, we have
\begin{itemize}
\item
The polynomial $p^\nu_{r-1,k_1-r-1,k_2,\ldots,k_n}$ has degree $(m-1)+(q+1)=m+q$.
\item
The polynomial $p^\nu_{k_1+k_2-2,k_3,\ldots,k_n}$ has degree $(m-1)+(q+1)=m+q$.
\item
The polynomial $p^\nu_{k_1+k_r-2,k_2,\ldots,\widehat{k_r},\ldots,k_n}$ has degree either:
\begin{itemize}
\item $(m+q)$ if $k_r$ is odd, or
\item $(m+q-2)$ if $k_r$ is even.
\end{itemize}
\end{itemize}
It follows that $p^\nu_{k_1,\ldots,k_n}$ has degree $(m+q)$.

\item {\it Case 2:} there is an index which is even and positive.

We assume that it is $k_1$ that is both even and positive. From \eqref{eqn_recursionbdry} we have
\begin{displaymath}
\begin{split}
p^\nu_{k_1,\ldots,k_n} =& \nu p^\nu_{k_1-2,k_2,\ldots,k_n} + \sum_{r=2}^{k_1-1} p^\nu_{r-1,k_1-r-1,k_2,\ldots,k_n} \\
&+ \sum_{r=2}^n k_r p^\nu_{k_1+k_r-2,k_2,\ldots,\widehat{k_r},\ldots,k_n}.
\end{split}
\end{displaymath}

Again, applying the inductive hypothesis, we have
\begin{itemize}
\item
The polynomial $\nu p^\nu_{k_1-2,k_2,\ldots,k_n}$ has degree $(m+q)$.
\item
The polynomial $p^\nu_{r-1,k_1-r-1,k_2,\ldots,k_n}$ has degree either:
\begin{itemize}
\item $(m+q)$ if $r$ is odd, or
\item $(m+q-2)$ if $r$ is even.
\end{itemize}
\item
The polynomial $p^\nu_{k_1+k_r-2,k_2,\ldots,\widehat{k_r},\ldots,k_n}$ has degree~$(m+q-2)$.
\end{itemize}
Again, it follows that $p^\nu_{k_1,\ldots,k_n}$ has degree~$(m+q)$.
\end{itemize}
This completes our computation of the degree of $p^\nu_{k_1,\ldots,k_n}$.
\end{proof}

Having now computed the degree of the polynomial $p^\nu_{k_1,\ldots,k_n}$,
we will look at describing some of the large $N$ asymptotic behavior of the polynomials in the next section.

\section{Large $N$ asymptotics of multi-trace correlation functions}
\label{sec_largeN}

In this section we will explain how to compute the leading, as well as some of the subleading coefficients of the polynomials $p^\nu_{k_1,\ldots,k_n}\in\mathbb{C}[\nu]$. By Proposition \ref{prop_poly-corfun} this will allow us to describe some of the large $N$ asymptotic behavior of the multi-trace correlation functions \eqref{eqn_corfunmultitrace}. In particular, we will use these results to compute the large $N$ statistical correlation coefficients between random variables in the Gaussian Unitary Ensemble and to prove a generalization of Wigner's semicircle law \cite{semicircle}.

\subsection{Computation of the leading coefficients} \label{sec_leadcoeff}

We begin by separating the cases of odd and even indices, using Theorem \ref{thm_degree} to compute the degrees of the polynomials.

\begin{defi}
Consider the polynomial
\[ p^\nu_{2i_1,2i_2,\ldots,2i_k}, \qquad\text{with } i_1,\ldots,i_k\geq 0 \]
having all even powers. It has degree $n_{i_1,\ldots,i_k}:=k+\sum_{r=1}^k i_r$. We use $C_{i_1,\ldots,i_k}$ to denote the coefficient of $\nu^{n_{i_1,\ldots,i_k}}$ for this polynomial.

Now consider the polynomial
\[ p^\nu_{2j_1+1,2j_2+1,\ldots,2j_{2l-1}+1,2j_{2l}+1}, \qquad\text{with } j_1,\ldots,j_{2l}\geq 0 \]
having all odd powers. It has degree $m_{j_1,\ldots,j_{2l}}:=l+\sum_{r=1}^{2l}j_r$. We use $A_{j_1,\ldots,j_{2l}}$ to denote the coefficient of $\nu^{m_{j_1,\ldots,j_{2l}}}$ for this polynomial.
\end{defi}

Now note that by Theorem \ref{thm_degree} the polynomial
\begin{equation} \label{eqn_generalpolynomial}
p^\nu_{2i_1,\ldots,2i_k,2j_1+1,\ldots,2j_{2l}+1}, \qquad\text{with } i_1,\ldots,i_k,j_1,\ldots,j_{2l}\geq 0
\end{equation}
has degree $n_{i_1,\ldots,i_k} + m_{j_1,\ldots,j_{2l}}$. The following result allows us to separate the cases of odd and even indices.

\begin{theorem} \label{thm_coeffoddeven}
The coefficient of $\nu^{n_{i_1,\ldots,i_k} + m_{j_1,\ldots,j_{2l}}}$ for the polynomial \eqref{eqn_generalpolynomial} is $C_{i_1,\ldots,i_k} A_{j_1,\ldots,j_{2l}}$.
\end{theorem}

\begin{cor} \label{cor_coeffoddeven}
For all $i_1,\ldots,i_k,j_1,\ldots,j_{2l}\geq 0$,
\[ \lim_{N\to\infty}\left[\frac{\int_{\her{N}}\tr (X^{2i_1})\cdots\tr (X^{2i_k}) \tr (X^{2j_1+1}) \cdots \tr (X^{2j_{2l}+1})e^{-\frac{1}{2}\tr (X^2)}\mathrm{d} X}{N^{n_{i_1,\ldots,i_k} + m_{j_1,\ldots,j_{2l}}}\int_{\her{N}}e^{-\frac{1}{2}\tr (X^2)}\mathrm{d} X}\right] = C_{i_1,\ldots,i_k} A_{j_1,\ldots,j_{2l}}. \]
\end{cor}

The proof of Theorem \ref{thm_coeffoddeven} will have to wait until we have at least computed the coefficients $C_{i_1,\ldots,i_k}$. Note that Corollary \ref{cor_coeffoddeven} is a consequence of Proposition~\ref{prop_poly-corfun}.

Recall that the $n$th \emph{Catalan number} is
\[ C_n=\frac{1}{n+1}\binom{2n}{n} = \frac{(2n)!}{n! (n+1)!}. \]
The following result thus explains our choice of notation for the leading coefficient of the polynomial~$p^\nu_{2i_1,\ldots,2i_k}$.

\begin{theorem}
\label{thm_catalanproduct}
For all $i_1,\ldots,i_k\geq 0$, the leading coefficient of $p^\nu_{2i_1,\ldots,2i_k}$ is a product of Catalan numbers:
\begin{equation} \label{eqn_catalanproduct}
C_{i_1,\ldots,i_k} = \prod_{r=1}^k \left[\frac{1}{i_r+1}\binom{2i_r}{i_r} \right] = \prod_{r=1}^k C_{i_r}.
\end{equation}
In particular, the leading coefficient of $p^\nu_{2n}$ is the $n$th Catalan number~$C_n$.
\end{theorem}

\begin{proof}
The proof is by induction on $\sum_{r=1}^k i_r$. The base case is trivial and occurs when all the indices $i_r$ are equal to zero.

Now, applying \eqref{eqn_recursionbdry} and assuming that $i_1\geq 1$ we have
\begin{equation} \label{eqn_recursionevencoeffs}
\begin{split}
p^\nu_{2i_1,\ldots,2i_k} =& \sum_{r=1}^{2i_1-1} p^\nu_{r-1,2i_1-r-1,2i_2,\ldots,2i_k} + \sum_{r=2}^k 2i_r p^\nu_{2(i_1+i_r-1),2i_2,\ldots,\widehat{2i_r},\ldots,2i_k} \\
=& \sum_{r=0}^{i_1-1} p^\nu_{2r,2(i_1-1-r),2i_2,\ldots,2i_k} + \sum_{r=1}^{i_1-1} p^\nu_{2r-1,2(i_1-r)-1,2i_2,\ldots,2i_k} \\
&+ \sum_{r=2}^k 2i_r p^\nu_{2(i_1+i_r-1),2i_2,\ldots,\widehat{2i_r},\ldots,2i_k}
\end{split}
\end{equation}
Note that by Theorem \ref{thm_degree}, only the first sum contains polynomials of the same degree as $p^\nu_{2i_1,\ldots,2i_k}$, with the remaining sums containing polynomials of lower degrees. This yields the recursion relation
\[ C_{i_1,\ldots,i_k} = \sum_{r=0}^{i_1-1} C_{r,i_1-1-r,i_2,\ldots,i_k}. \]
Upon substituting \eqref{eqn_catalanproduct}, this recursion relation becomes the well-known and canonical recursion relation
\begin{equation} \label{eqn_catalanrecursion}
C_{n+1} = \sum_{i=0}^n C_i C_{n-i}
\end{equation}
for the Catalan numbers.
\end{proof}

We may now prove Theorem \ref{thm_coeffoddeven}.

\begin{proof}[Proof of Theorem \ref{thm_coeffoddeven}]
To prove Theorem \ref{thm_coeffoddeven} we proceed by induction on
\[ l + \sum_{r=1}^{2l}j_r + \sum_{r=1}^k i_r \geq 1. \]
The base case is $i_1=\ldots=i_k=j_1=\ldots=j_{2l}=0$ and $l=1$, in which case the polynomial \eqref{eqn_generalpolynomial} is $\nu^k p^\nu_{1,1} = \nu^{k+1}$; hence the base case is clearly satisfied.

Applying \eqref{eqn_recursionbdry} to the polynomial $p^\nu_{2i_1,\ldots,2i_k,2j_1+1,\ldots,2j_{2l}+1} = p^\nu_{2j_1+1,\ldots,2j_{2l}+1,2i_1,\ldots,2i_k}$ we get
\begin{equation} \label{eqn_oddevenrecursion}
\begin{split}
p^\nu_{2j_1+1,\ldots,2j_{2l}+1,2i_1,\ldots,2i_k} =&
2\sum_{r=1}^{j_1} p^\nu_{2(r-1),2(j_1-r)+1,2j_2+1,\ldots,2j_{2l}+1,2i_1,\ldots,2i_k} \\
&+ \sum_{r=2}^{2l} (2j_r+1)p^\nu_{2(j_1+j_r),2j_2+1,\ldots,\widehat{2j_r+1},\ldots,2j_{2l}+1,2i_1,\ldots,2i_k} \\
&+ \sum_{r=1}^k 2i_r p^\nu_{2(j_1+i_r)-1,2j_2+1,\ldots,2j_{2l}+1,2i_1,\ldots,\widehat{2i_r},\ldots,2i_k}.
\end{split}
\end{equation}
Here we have split up the first sum in \eqref{eqn_recursionbdry} into sums over even and odd indices and then identified these two sums. Note that by Theorem \ref{thm_degree} the last sum in \eqref{eqn_oddevenrecursion} contains polynomials of strictly smaller degree and hence may be ignored for our purposes.

Applying the inductive hypothesis to \eqref{eqn_oddevenrecursion} we conclude that the leading coefficient of \eqref{eqn_generalpolynomial} is
\begin{multline} \label{eqn_oddevenrecursioncoeffs}
2\sum_{r=1}^{j_1} C_{r-1,i_1,\ldots,i_k} A_{j_1-r,j_2,\ldots,j_{2l}} + \sum_{r=2}^{2l} (2j_r+1) C_{j_1+j_r,i_1,\ldots,i_k} A_{j_2,\ldots,\widehat{j_r},\ldots,j_{2l}} \\
= C_{i_1,\ldots,i_k}\left(2\sum_{r=1}^{j_1} C_{r-1} A_{j_1-r,j_2,\ldots,j_{2l}} + \sum_{r=2}^{2l} (2j_r+1) C_{j_1+j_r} A_{j_2,\ldots,\widehat{j_r},\ldots,j_{2l}}\right)
\end{multline}
where on the last line of \eqref{eqn_oddevenrecursioncoeffs} we have applied Equation \eqref{eqn_catalanproduct}.

In particular, setting $k=0$ in \eqref{eqn_oddevenrecursion} and applying the inductive hypothesis yields the identity
\begin{equation} \label{eqn_recursionoddcoeffs}
A_{j_1,\ldots,j_{2l}} = 2\sum_{r=1}^{j_1} C_{r-1} A_{j_1-r,j_2,\ldots,j_{2l}} + \sum_{r=2}^{2l}(2j_r+1) C_{j_1+j_r} A_{j_2,\ldots,\widehat{j_r},\ldots,j_{2l}}.
\end{equation}
Substituting \eqref{eqn_recursionoddcoeffs} into \eqref{eqn_oddevenrecursioncoeffs} we conclude that the leading coefficient of \eqref{eqn_generalpolynomial} is~$C_{i_1,\ldots,i_k} A_{j_1,\ldots,j_{2l}}$.
\end{proof}

It remains to compute the coefficient $A_{j_1,\ldots,j_{2l}}$, which we will do using the recurrence relation~\eqref{eqn_recursionoddcoeffs}.
We begin by examining the simplest possible case and derive an expression in terms of the Catalan numbers~$C_n$.

\begin{lemma} \label{lem_oddpaircoeff}
For all $j_1,j_2\geq 0$, we have
\begin{equation} \label{eqn_oddpaircoeff}
A_{j_1,j_2} = (2j_2+1)\sum_{r=0}^{j_1} (r+1)C_r C_{j_1+j_2-r}.
\end{equation}
\end{lemma}

\begin{proof}
Consider the generating function $c(t)=\sum_{n=0}^\infty C_n t^n$ for the Catalan numbers and recall that the canonical recursion relation~\eqref{eqn_catalanrecursion} is encoded as
\begin{equation} \label{eqn_catalanrecursiongenfun}
c(t) = 1+tc(t)^2.
\end{equation}
Now introduce the generating functions
\[ a_{j_2}(t):=\sum_{n=0}^\infty A_{n,j_2}t^n \quad\text{and}\quad \vec{c}_{j_2}(t):=\sum_{n=0}^\infty C_{n+j_2}t^n. \]
From Theorem \ref{thm_coeffoddeven} and the recursion relation~\eqref{eqn_recursionoddcoeffs}, we obtain the recursion
\begin{equation} \label{eqn_recursionoddcoeffspair}
A_{j_1,j_2} = 2\sum_{r=0}^{j_1 - 1} C_r A_{j_1-1-r,j_2} + (2j_2+1) C_{j_1+j_2}.
\end{equation}
In terms of these generating functions, the recursion relation \eqref{eqn_recursionoddcoeffspair} can be written as
\[ a_{j_2}(t) = (2j_2+1)\vec{c}_{j_2}(t) + 2tc(t)a_{j_2}(t) \]
from which we get
\begin{equation} \label{eqn_oddpairgenfunsoln}
a_{j_2}(t) = (2j_2+1)(1-2tc(t))^{-1}\vec{c}_{j_2}(t).
\end{equation}

One approach to computing the power series $(1-2tc(t))^{-1}$ is to differentiate \eqref{eqn_catalanrecursiongenfun}; however, we choose to proceed as follows. Completing the square in \eqref{eqn_catalanrecursiongenfun} we obtain
\begin{equation} \label{eqn_genfuninverse}
\begin{split}
(1-2tc(t))^2 &= 1-4t \\
(1-2tc(t))^{-1} &= (1-4t)^{-\frac{1}{2}} = \sum_{n=0}^\infty \binom{-\frac{1}{2}}{n}(-4t)^n \\
&= \sum_{n=0}^{\infty}\binom{2n}{n}t^n = \sum_{n=0}^\infty (n+1)C_nt^n,
\end{split}
\end{equation}
which one may observe is the derivative of $tc(t)$. Substituting the above into \eqref{eqn_oddpairgenfunsoln} we obtain the equation~\eqref{eqn_oddpaircoeff}.
\end{proof}

\begin{rem}
Note that the coefficient $A_{j_1,j_2}$ is symmetric in $j_1$ and $j_2$, whilst there appears to be no (alternative) {\it a priori} explanation for why the expression on the right-hand side of \eqref{eqn_oddpaircoeff} possesses this symmetry.
One might view this as an indication that the formula \eqref{eqn_oddpaircoeff} is not canonical, and we will indeed derive a much cleaner expression for this coefficient in Theorem~\ref{thm_coeffsublead}.
\end{rem}

The final step we must take is to derive an expression for $A_{j_1,\ldots,j_{2l}}$ in terms of these simpler coefficients. Given a chord diagram
\[ c=\{r_1,s_1\},\{r_2,s_2\},\ldots,\{r_l,s_l\}\in\chord{2l} \]
define
\[ A_{j_1,\ldots,j_{2l}}(c) := A_{j_{r_1},j_{s_1}}A_{j_{r_2},j_{s_2}}\cdots A_{j_{r_l},j_{s_l}}. \]

\begin{prop} \label{prop_oddcoeff}
For all $j_1,\ldots,j_{2l}\geq 0$,
\[ A_{j_1,\ldots,j_{2l}} = \sum_{c\in\chord{2l}} A_{j_1,\ldots,j_{2l}}(c). \]
\end{prop}

\begin{proof}
The proof proceeds by induction on $\sum_{r=1}^{2l} j_r$. The base case occurs when all the indices $j_r$ are equal to zero and hence follows from the calculation performed in Example \ref{exm_allones} and the fact that there are $(2l-1)!!$ chord diagrams in $\chord{2l}$.

Now, assuming that $j_1\geq 1$, we apply the recurrence relation \eqref{eqn_recursionoddcoeffs} and the inductive hypothesis to conclude that
\[ A_{j_1,\ldots,j_{2l}} = 2\sum_{i=1}^{j_1}\sum_{c\in\chord{2l}} C_{i-1} A_{j_1-i,j_2,\ldots,j_{2l}}(c) + \sum_{k=2}^{2l}\sum_{c\in\chord{l-1}}(2j_k+1) C_{j_1+j_k} A_{j_2,\ldots,\widehat{j_k},\ldots,j_{2l}}(c). \]
Using the identity
\begin{equation} \label{eqn_chordidentity}
\sum_{c\in\chord{2l}} A_{j_1,\ldots,j_{2l}}(c) = \sum_{k=2}^{2l}\sum_{c\in\chord{l-1}} A_{j_1,j_k}A_{j_2,\ldots,\widehat{j_k},\ldots,j_{2l}}(c)
\end{equation}
and the recurrence relation \eqref{eqn_recursionoddcoeffspair} we obtain
\begin{displaymath}
\begin{split}
\sum_{c\in\chord{2l}} A_{j_1,\ldots,j_{2l}}(c) =&  2\sum_{k=2}^{2l}\sum_{c\in\chord{l-1}}\sum_{i=1}^{j_1} C_{i-1} A_{j_1-i,j_k}A_{j_2,\ldots,\widehat{j_k},\ldots,j_{2l}}(c) \\
&+ \sum_{k=2}^{2l}\sum_{c\in\chord{l-1}}(2j_k+1) C_{j_1+j_k}A_{j_2,\ldots,\widehat{j_k},\ldots,j_{2l}}(c) \\
=&  2\sum_{i=1}^{j_1}\sum_{c\in\chord{2l}} C_{i-1} A_{j_1-i,j_2,\ldots,j_{2l}}(c) \\
&+ \sum_{k=2}^{2l}\sum_{c\in\chord{l-1}}(2j_k+1) C_{j_1+j_k}A_{j_2,\ldots,\widehat{j_k},\ldots,j_{2l}}(c) = A_{j_1,\ldots,j_{2l}}.
\end{split}
\end{displaymath}
\end{proof}

\begin{rem}
From Proposition \ref{prop_oddcoeff} and Equation \eqref{eqn_chordidentity} we obtain the useful recurrence relation
\[ A_{j_1,\ldots,j_{2l}} = \sum_{k=2}^{2l} A_{j_1,j_k}A_{j_2,\ldots,\widehat{j_k},\ldots,j_{2l}}. \]
\end{rem}

This concludes the first step of calculating of the leading coefficient of \eqref{eqn_generalpolynomial}. In the meantime, as a simple demonstration of our results, we provide the following example.

\begin{example}
\begin{multline*}
\lim_{N\to\infty}\left[\frac{\int_{\her{N}}\tr (X^{10})\tr (X^{42}) \tr (X^{15})\tr (X^{43})\tr (X^{47})\tr (X^{63})e^{-\frac{1}{2}\tr (X^2)}\mathrm{d} X}{N^{112}\int_{\her{N}}e^{-\frac{1}{2}\tr (X^2)}\mathrm{d} X}\right] = C_{5,21} A_{7,21,23,31} \\ =25081904924688737847061935982290890890757044619026344345600000.
\end{multline*}
The details are left, of course, to the enthusiastic reader.
\end{example}

\subsection{Subleading coefficients}
\label{sec_subleadcoeff}

In order to compute the large $N$ statistical correlations between certain random variables in the Gaussian Unitary Ensemble in Section \ref{sec_largeNcorcoeff}, it is necessary to supplement our knowledge of the large $N$ asymptotic behavior of the correlation functions \eqref{eqn_corfunmultitrace} that we obtained in the preceding section by computing the coefficient of the subleading term in $p^\nu_{2i_1,\ldots,2i_k}$. This will also lead us to a cleaner formula for the coefficient $A_{j_1,\ldots,j_{2l}}$.

Recall that the polynomial $p^\nu_{2i_1,\ldots,2i_k}$ is a polynomial of degree $n:=k+\sum_{s=1}^k i_s$. Denote the coefficient of $\nu^{n-2r}$ in $p^\nu_{2i_1,\ldots,2i_k}$ by $C_{i_1,\ldots,i_k}(r)$ so that in the notation of the preceding section we have
\[ C_{i_1,\ldots,i_k} = C_{i_1,\ldots,i_k}(0). \]
We will compute $C_{i_1,\ldots,i_k}(1)$. To do so we make the auxiliary definition
\begin{equation} \label{eqn_auxsublead}
\widetilde{C}_{i,j}(1) := C_{i,j}(1) - C_i(1)C_j - C_j(1)C_i
\end{equation}
which is symmetric with respect to $i$ and $j$.

\begin{prop} \label{prop_subleading}
For all $i_1,\ldots,i_k\geq 0$,
\begin{equation} \label{eqn_subleading}
C_{i_1,\ldots,i_k}(1) = \sum_{r=1}^k C_{i_r}(1)C_{i_1,\ldots,\widehat{i_r},\ldots,i_k} + \sum_{1\leq r<s \leq k} \widetilde{C}_{i_r,i_s}(1)C_{i_1,\ldots,\widehat{i_r},\ldots,\widehat{i_s},\ldots,i_k}.
\end{equation}
\end{prop}

\begin{rem}
It follows from Equation \eqref{eqn_subleading} that
\begin{equation} \label{eqn_subleadingsplit}
\begin{split}
C_{i_1,\ldots,i_k,i'_1,\ldots,i'_{k'}}(1) =& C_{i_1,\ldots,i_k}(1)C_{i'_1,\ldots,i'_{k'}} + C_{i'_1,\ldots,i'_{k'}}(1)C_{i_1,\ldots,i_k} \\
&+ \sum_{r=1}^k\sum_{s=1}^{k'}\widetilde{C}_{i_r,i'_s}(1)C_{i_1,\ldots,\widehat{i_r},\ldots,i_k}C_{i'_1,\ldots,\widehat{i'_s},\ldots,i'_{k'}}
\end{split}
\end{equation}
for all $i_1,\ldots,i_k;i'_1,\ldots,i'_{k'}\geq 0$.
\end{rem}

\begin{proof}[Proof of Proposition \ref{prop_subleading}]
The proof is by induction on $\sum_{r=1}^k i_r$. The base case occurs when all the indices $i_r$ are equal to zero, in which case both sides of \eqref{eqn_subleading} clearly vanish.

We note that from Equation \eqref{eqn_recursionevencoeffs}, Theorem \ref{thm_degree} and Theorem \ref{thm_coeffoddeven} we obtain the recursion relation
\begin{equation} \label{eqn_subleadingrecursionproxy}
C_{i_1,\ldots,i_k}(1) = \sum_{r=0}^{i_1-1}C_{r,i_1-1-r,i_2,\ldots,i_k}(1) + \sum_{r=0}^{i_1-2}A_{r,i_1-2-r}C_{i_2,\ldots,i_k} + 2\sum_{r=2}^k i_rC_{i_1+i_r-1,i_2,\ldots,\widehat{i_r},\ldots,i_k}
\end{equation}
which is valid for all $i_1\geq 1$ and $i_2,\ldots,i_k\geq 0$. In particular, setting $k=1$ in \eqref{eqn_subleadingrecursionproxy} we obtain
\begin{equation} \label{eqn_subleadingrecursion1pt}
C_i(1) = \sum_{r=0}^{i-1}C_{r,i-1-r}(1) + \sum_{r=0}^{i-2}A_{r,i-2-r}
\end{equation}
which is valid for all $i\geq 1$. Combining \eqref{eqn_subleadingrecursion1pt} with \eqref{eqn_subleadingrecursionproxy} yields the recursion
\begin{equation} \label{eqn_subleadingrecursion}
\begin{split}
C_{i_1,\ldots,i_k}(1) =& \sum_{r=0}^{i_1-1}C_{r,i_1-1-r,i_2,\ldots,i_k}(1) + C_{i_2,\ldots,i_k}\left(C_{i_1}(1) - \sum_{r=0}^{i_1-1}C_{r,i_1-1-r}(1)\right) \\
&+ 2\sum_{r=2}^k i_rC_{i_1+i_r-1,i_2,\ldots,\widehat{i_r},\ldots,i_k}
\end{split}
\end{equation}
which is valid for all $i_1\geq 1$ and $i_2,\ldots,i_k\geq 0$.

Applying the inductive hypothesis and making use of Equation \eqref{eqn_subleadingsplit} we find that
\begin{displaymath}
\begin{split}
C_{r,i_1-1-r,i_2,\ldots,i_k}(1) =& C_{r,i_1-1-r}(1)C_{i_2,\ldots,i_k} + C_{i_2,\ldots,i_k}(1)C_{r,i_1-1-r} \\
&+ \sum_{s=2}^k \widetilde{C}_{r,i_s}(1)C_{i_1-1-r,i_2,\ldots,\widehat{i_s},\ldots,i_k} + \sum_{s=2}^k \widetilde{C}_{i_1-1-r,i_s}(1)C_{r,i_2,\ldots,\widehat{i_s},\ldots,i_k}.
\end{split}
\end{displaymath}
Substituting the above into \eqref{eqn_subleadingrecursion} whilst making use of \eqref{eqn_catalanrecursion} and \eqref{eqn_catalanproduct} we obtain
\begin{equation} \label{eqn_subleadinginduction}
\begin{split}
C_{i_1,\ldots,i_k}(1) =& C_{i_2,\ldots,i_k}(1)\sum_{r=0}^{i_1-1}C_r C_{i_1-1-r} + 2\sum_{s=2}^k\sum_{r=0}^{i_1-1} \widetilde{C}_{r,i_s}(1)C_{i_1-1-r,i_2,\ldots,\widehat{i_s},\ldots,i_k} \\
&+ C_{i_2,\ldots,i_k}C_{i_1}(1) + 2\sum_{s=2}^k i_sC_{i_1+i_s-1,i_2,\ldots,\widehat{i_s},\ldots,i_k} \\
=& C_{i_2,\ldots,i_k}(1)C_{i_1} + C_{i_1}(1)C_{i_2,\ldots,i_k} \\
&+ 2\sum_{s=2}^k C_{i_2,\ldots,\widehat{i_s},\ldots,i_k}\left(i_sC_{i_1+i_s-1} + \sum_{r=0}^{i_1-1} \widetilde{C}_{r,i_s}(1)C_{i_1-1-r}\right).
\end{split}
\end{equation}

Repeating the argument above once more we obtain the formula
\[ C_{i_1,i_s}(1) = C_{i_1}(1)C_{i_s} + C_{i_s}(1)C_{i_1} + 2\left(i_sC_{i_1+i_s-1} + \sum_{r=0}^{i_1-1}\widetilde{C}_{r,i_s}(1)C_{i_1-1-r}\right). \]
Substituting the above into \eqref{eqn_subleadinginduction} and using \eqref{eqn_auxsublead} we arrive at the equation
\[ C_{i_1,\ldots,i_k}(1) = C_{i_2,\ldots,i_k}(1)C_{i_1} + C_{i_1}(1)C_{i_2,\ldots,i_k} + \sum_{s=2}^k C_{i_2,\ldots,\widehat{i_s},\ldots,i_k}\widetilde{C}_{i_1,i_s}(1). \]
Applying the inductive hypothesis once more to the first term in this equation finally establishes Equation~\eqref{eqn_subleading}.
\end{proof}

\begin{rem}
A formula for the coefficient $C_i(1)$ may be found using the Harer-Zagier recursion relation \cite{HZ}. Since we will not actually require such a formula for our purposes, we will not provide the full details, but it may be computed to be
\begin{displaymath}
C_i(1) = \left\{\begin{array}{cl}\binom{2i-1}{3}C_{i-2}, & i\geq 2 \\ 0, & i=0,1 \end{array}\right\}.
\end{displaymath}
\end{rem}

It therefore remains to compute the terms $\widetilde{C}_{i,j}(1)$. This is done in a similar fashion to our computation of the coefficients $A_{i,j}$ in Lemma \ref{lem_oddpaircoeff}. One consequence of this analysis will be that we will also obtain a very clean formula for the coefficients $A_{i,j}$.

\begin{theorem} \label{thm_coeffsublead}
For all $i,j\geq 0$,
\begin{align}
\label{eqn_coeffsublead}
\widetilde{C}_{i,j}(1) &= \left\{\begin{array}{cl}\frac{ij}{(i+j)}\binom{2i}{i}\binom{2j}{j}, & i+j>0 \\ 0, & i+j=0\end{array}\right\}, \\
\label{eqn_coeffoddpairlead}
A_{i,j} &= \frac{(2i+1)!(2j+1)!}{(i+j+1)(i!j!)^2} = \frac{(2i+1)(2j+1)}{(i+j+1)}\binom{2i}{i}\binom{2j}{j}.
\end{align}
\end{theorem}

\begin{proof}
From the recursion relation \eqref{eqn_subleadingrecursion} and formula \eqref{eqn_subleadingsplit} we obtain Equation \eqref{eqn_subleadinginduction}, which for two indices $i\geq 1$ and $j\geq 0$ provides us with the recursion relation
\begin{equation} \label{eqn_subleadingrecursion2pt}
\widetilde{C}_{i,j}(1) = 2\left(jC_{i+j-1}+\sum_{r=0}^{i-1}\widetilde{C}_{r,j}(1)C_{i-1-r}\right).
\end{equation}

Now consider the generating functions
\[ c(t):=\sum_{n=0}^\infty C_n t^n, \quad \tilde{c}_j(t):=\sum_{n=0}^\infty\widetilde{C}_{n,j}(1) t^n, \quad\text{and}\quad \vec{c}_j(t):=\sum_{n=0}^\infty C_{n+j}t^n. \]
The recursion relation \eqref{eqn_subleadingrecursion2pt} may be written in terms of these generating functions as
\[ \tilde{c}_j(t) = 2t\left(c(t)\tilde{c}_j(t)+j\vec{c}_j(t)\right). \]
Rearranging and using our expression \eqref{eqn_genfuninverse} for $(1-2tc(t))^{-1}$ we obtain
\begin{displaymath}
\begin{split}
\tilde{c}_j(t) &= 2j(1-2tc(t))^{-1}t\vec{c}_j(t) \\
&= 2j\sum_{n=1}^\infty\left(\sum_{r=0}^{n-1}(r+1)C_rC_{n+j-r-1}\right)t^n.
\end{split}
\end{displaymath}
Therefore,
\begin{equation} \label{eqn_solvesublead}
\widetilde{C}_{i,j}(1) = 2j\sum_{r=0}^{i-1}(r+1)C_rC_{i+j-r-1} = \frac{2j}{2j+1} A_{i-1,j}
\end{equation}
where at the last step we have used Equation \eqref{eqn_oddpaircoeff}.

Using the fact that the coefficients $A_{i,j}$ and $\widetilde{C}_{i,j}(1)$ are symmetric in $i$ and $j$, we obtain from \eqref{eqn_solvesublead} the recursion relation
\begin{equation} \label{eqn_oddcoeffpairrecursion}
\begin{array}{ll}
A_{i,j-1} = \frac{j(2i+1)}{i(2j+1)}A_{i-1,j}, & i,j\geq 1 \\
A_{0,j} = (2j+1)C_j, & j\geq 0
\end{array}
\end{equation}
where we have used Equation \eqref{eqn_oddpaircoeff} to compute $A_{0,j}$.
The recursive computation \eqref{eqn_oddcoeffpairrecursion} may be solved in a straightforward manner and it is easily verified that the solution is provided by Equation \eqref{eqn_coeffoddpairlead}.
Equation \eqref{eqn_coeffsublead} then follows from \eqref{eqn_coeffoddpairlead} and Equation~\eqref{eqn_solvesublead}.
\end{proof}

With this new formula for the coefficients $A_{i,j}$, we can use Proposition~\ref{prop_oddcoeff} to give a more convenient expression for the coefficients $A_{j_1,\ldots,j_{2l}}$.
Given a chord diagram
\[ c=\{r_1,s_1\},\{r_2,s_2\},\ldots,\{r_l,s_l\}\in\chord{2l} \]
define
\[ \mu_{j_1,\ldots,j_{2l}}(c):=\prod_{k=1}^l \left[\frac{1}{j_{r_k}+j_{s_k}+1}\right] \]
and
\begin{equation} \label{eqn_chordcoeffdef}
\mu_{j_1,\ldots,j_{2l}} := \sum_{c\in\chord{2l}}\mu_{j_1,\ldots,j_{2l}}(c)
\end{equation}
for all $j_1,\ldots,j_{2l}\geq 0$. Note again the convenient formula
\[ \mu_{j_1,\ldots,j_{2l}} = \sum_{k=2}^{2l}\mu_{j_1,j_k}\mu_{j_2,\ldots,\widehat{j_k},\ldots,j_{2l}}. \]

\begin{theorem}
For all $j_1,\ldots,j_{2l}\geq 0$,
\begin{equation} \label{eqn_oddcoeffsimplified}
A_{j_1,\ldots,j_{2l}} = \mu_{j_1,\ldots,j_{2l}}\prod_{k=1}^{2l}\left[\frac{(2j_k+1)!}{(j_k!)^2}\right].
\end{equation}
\end{theorem}

\begin{proof}
This is an immediate consequence of Proposition \ref{prop_oddcoeff} and Theorem~\ref{thm_coeffsublead}.
\end{proof}

\subsection{Large $N$ statistical correlation coefficients in the Gaussian Unitary Ensemble} \label{sec_largeNcorcoeff}

Consider the algebra of conjugation-invariant \textit{polynomial} functions on matrices,
\[\sym{\mat{N}{\mathbb{C}}^*}^{\gl{n}{\mathbb{C}}}.\]
By the Fundamental Theorem of Invariant Theory for $\GL{N}{\mathbb{C}}$, we know that this algebra is generated by traces~$\tr(X^k)$ of powers of the matrix~$X$. We can see these functions, of course, as random variables in the Gaussian Unitary Ensemble. In terms of this paper, we note that
\[ \sym{\mat{N}{\mathbb{C}}^*}^{\gl{n}{\mathbb{C}}}\subset\BV{\mat{N}{\twodim}}, \]
so that our approach has full access to this class of random variables.

Let $f$ and $g$ be two conjugation-invariant, polynomial random variables .
Recall that the variance, covariance and statistical correlation between the random variables $f$ and $g$ are defined by
\begin{displaymath}
\begin{split}
\var{f} &:= \langle f^2 \rangle - \langle f \rangle^2 \\
\cov{f}{g} &:= \langle fg \rangle - \langle f \rangle\langle g \rangle \\
\corr{f}{g} &:= \frac{\cov{f}{g}}{\sqrt{\var{f}\var{g}}}
\end{split}
\end{displaymath}
where we remind the reader that the expected value $\langle f \rangle$ of the random variable $f$ on the sample space of Hermitian matrices is defined by the expectation value map~\eqref{eqn_expvalmap}.

By Proposition \ref{prop_poly-corfun} the expected value of the random variable
\begin{equation} \label{eqn_ranvargen}
f(X) := \tr(X^{k_1})\tr(X^{k_2})\cdots\tr(X^{k_n}), \quad X\in\her{N}
\end{equation}
is
\[\langle f \rangle = p^\nu_{k_1,\ldots,k_n}(N).\]
In this section we will compute the large $N$ statistical correlations between the random variables~\eqref{eqn_ranvargen}; i.e., between multi-trace functions. We will find that the resulting expressions are dominated by the presence of odd powers in \eqref{eqn_ranvargen}.

\begin{theorem} \label{thm_largeNcorrelation}
Fix nonnegative integers
\[ i_1,\ldots,i_k;j_1,\ldots,j_l;i'_1,\ldots,i'_{k'};j'_1,\ldots,j'_{l'}\geq 0 \quad\text{with}\quad k,l,k',l'\geq 0. \]
These determine sequences of random variables $(f_N)$ and $(g_N)$, where for each $N\in\mathbb{N}$, we have functions on the rank $N$ Gaussian Unitary Ensemble:
for $X \in \her{N}$,
\begin{displaymath}
\begin{split}
f_N(X) &:= \tr(X^{2i_1})\cdots\tr(X^{2i_k})\tr(X^{2j_1+1})\cdots\tr(X^{2j_l+1}), \\
g_N(X) &:= \tr(X^{2i'_1})\cdots\tr(X^{2i'_{k'}})\tr(X^{2j'_1+1})\cdots\tr(X^{2j'_{l'}+1}).
\end{split}
\end{displaymath}
Then the following are true:
\begin{enumerate}
\item
If $l$ and $l'$ have opposite parities, then
\[ \corr{f_N}{g_N} = 0 \]
for all $N\in\mathbb{N}$.
\item
If $l$ and $l'$ are both odd, then
\begin{equation} \label{eqn_largeNcorrelation_2}
\lim_{N\to\infty}\corr{f_N}{g_N} = \frac{\mu_{j_1,\ldots,j_l,j'_1,\ldots,j'_{l'}}}{\sqrt{\mu_{j_1,\ldots,j_l,j_1,\ldots,j_l}\mu_{j'_1,\ldots,j'_{l'},j'_1,\ldots,j'_{l'}}}}.
\end{equation}
\item
If $l$ and $l'$ are both even and $l,l'>0$, then
\begin{equation} \label{eqn_largeNcorrelation_3}
\lim_{N\to\infty}\corr{f_N}{g_N} = \frac{\mu_{j_1,\ldots,j_l,j'_1,\ldots,j'_{l'}} - \mu_{j_1,\ldots,j_l}\mu_{j'_1,\ldots,j'_{l'}}}{\sqrt{\mu_{j_1,\ldots,j_l,j_1,\ldots,j_l} - \mu_{j_1,\ldots,j_l}^2}\sqrt{\mu_{j'_1,\ldots,j'_{l'},j'_1,\ldots,j'_{l'}} - \mu_{j'_1,\ldots,j'_{l'}}^2}}.
\end{equation}
\item
If $l$ is even, $l>0$, and $l'=0$, then (assuming not all $i'_r=0$)
\begin{equation} \label{eqn_largeNcorrelation_4}
\lim_{N\to\infty}\corr{f_N}{g_N} = 0.
\end{equation}
\item
If $l=l'=0$, then (assuming that $i_1,\ldots,i_k,i'_1,\ldots,i'_{k'}>0$)
\begin{equation} \label{eqn_largeNcorrelation_5}
\lim_{N\to\infty}\corr{f_N}{g_N} = \frac{\sum_{r=1}^k\sum_{s=1}^{k'}\frac{i_ri'_s(i_r+1)(i'_s+1)}{(i_r+i'_s)}}{\sqrt{\sum_{r,s=1}^k\frac{i_ri_s(i_r+1)(i_s+1)}{(i_r+i_s)}}\sqrt{\sum_{r,s=1}^{k'}\frac{i'_ri'_s(i'_r+1)(i'_s+1)}{(i'_r+i'_s)}}}.
\end{equation}
\end{enumerate}
\end{theorem}

\begin{rem}
Note that for Equation \eqref{eqn_largeNcorrelation_5} the assumption that the indices $i_1,\ldots,i_k$ and $i'_1,\ldots,i'_{k'}$ are all positive is not a real constraint, since the inclusion of indices that are equal to zero will not affect the correlation coefficient. This is because the statistical correlation between random variables does not change after multiplying the random variables by a scalar.
\end{rem}

\begin{rem}
For convenience we point out that in the simplest instances the above formulas provide the following nice expressions:
\begin{displaymath}
\begin{split}
\lim_{N\to\infty}\corr{\tr(X^{2i+1})}{\tr(X^{2j+1})} &= \frac{\sqrt{(2i+1)(2j+1)}}{(i+j+1)}, \\
\lim_{N\to\infty}\corr{\tr(X^{2i})}{\tr(X^{2j})} &= \frac{2\sqrt{ij}}{(i+j)}.
\end{split}
\end{displaymath}
\end{rem}

\begin{proof}[Proof of Theorem \ref{thm_largeNcorrelation}]
\
\begin{enumerate}
\item
In this case $\langle f_N g_N \rangle = \langle f_N \rangle\langle g_N \rangle = 0$ for all positive integers $N$ and hence the covariance $\cov{f_N}{g_N}$ vanishes for all $N$.
\item
In this case we have $\langle f_N \rangle = \langle g_N \rangle = 0$ for all $N$. Therefore, using Theorem \ref{thm_coeffoddeven} and Theorem \ref{thm_degree} we find that
\begin{displaymath}
\begin{split}
\cov{f_N}{g_N} &= \langle f_N g_N \rangle = p^\nu_{2i_1,\ldots,2i_k,2i'_1,\ldots,2i'_{k'},2j_1+1,\ldots,2j_l+1,2j'_1+1,\ldots,2j'_{l'}+1}(N) \\
&= C_{i_1,\ldots,i_k,i'_1,\ldots,i'_{k'}}A_{j_1,\ldots,j_l,j'_1,\ldots,j'_{l'}}N^n + O(N^{n-2})
\end{split}
\end{displaymath}
where
\begin{equation} \label{eqn_corrdeg}
n := k + k' + \frac{(l+l')}{2} + \sum_{r=1}^k i_r + \sum_{r=1}^{k'} i'_r + \sum_{r=1}^l j_r + \sum_{r=1}^{l'} j'_r.
\end{equation}

We find similar expressions for the variances:
\begin{displaymath}
\begin{split}
\var{f_N} &= \cov{f_N}{f_N} = C_{i_1,\ldots,i_k,i_1,\ldots,i_k}A_{j_1,\ldots,j_l,j_1,\ldots,j_l}N^{2n_f} + O(N^{2(n_f-1)}), \\
\var{g_N} &= C_{i'_1,\ldots,i'_{k'},i'_1,\ldots,i'_{k'}}A_{j'_1,\ldots,j'_{l'},j'_1,\ldots,j'_{l'}}N^{2n_g} + O(N^{2(n_g-1)});
\end{split}
\end{displaymath}
where
\begin{equation} \label{eqn_corrdegs}
n_f := k + \frac{l}{2} + \sum_{r=1}^k i_r + \sum_{r=1}^l j_r \quad\text{and}\quad n_g := k' + \frac{l'}{2} + \sum_{r=1}^{k'} i'_r + \sum_{r=1}^{l'} j'_r
\end{equation}
so that $n=n_f+n_g$.

Putting these together we find that
\[ \corr{f_N}{g_N} = \frac{C_{i_1,\ldots,i_k}C_{i'_1,\ldots,i'_{k'}}A_{j_1,\ldots,j_l,j'_1,\ldots,j'_{l'}} +O(\frac{1}{N^2})}{\sqrt{C_{i_1,\ldots,i_k}^2 A_{j_1,\ldots,j_l,j_1,\ldots,j_l}+O(\frac{1}{N^2})}\sqrt{C_{i'_1,\ldots,i'_{k'}}^2 A_{j'_1,\ldots,j'_{l'},j'_1,\ldots,j'_{l'}}+O(\frac{1}{N^2})}} \]
and hence that
\[ \lim_{N\to\infty}\corr{f_N}{g_N} = \frac{A_{j_1,\ldots,j_l,j'_1,\ldots,j'_{l'}} }{\sqrt{A_{j_1,\ldots,j_l,j_1,\ldots,j_l}}\sqrt{A_{j'_1,\ldots,j'_{l'},j'_1,\ldots,j'_{l'}}}}. \]
Using the formula \eqref{eqn_oddcoeffsimplified} for the coefficients in the numerator and the denominator of the above expression we find that all the factorials cancel, leaving us with Equation~\eqref{eqn_largeNcorrelation_2}.

\item
Proceeding as above, we find that
\begin{displaymath}
\begin{split}
\cov{f_N}{g_N} =& p^\nu_{2i_1,\ldots,2i_k,2i'_1,\ldots,2i'_{k'},2j_1+1,\ldots,2j_l+1,2j'_1+1,\ldots,2j'_{l'}+1}(N) \\
&-p^\nu_{2i_1,\ldots,2i_k,2j_1+1,\ldots,2j_l+1}(N)p^\nu_{2i'_1,\ldots,2i'_{k'},2j'_1+1,\ldots,2j'_{l'}+1}(N) \\
=& C_{i_1,\ldots,i_k,i'_1,\ldots,i'_{k'}}A_{j_1,\ldots,j_l,j'_1,\ldots,j'_{l'}}N^n + O(N^{n-2}) \\
&-\big{(}C_{i_1,\ldots,i_k}A_{j_1,\ldots,j_l}N^{n_f} + O(N^{n_f-2})\big{)} \big{(}C_{i'_1,\ldots,i'_{k'}}A_{j'_1,\ldots,j'_{l'}}N^{n_g} + O(N^{n_g-2})\big{)} \\
=& C_{i_1,\ldots,i_k}C_{i'_1,\ldots,i'_{k'}}\left(A_{j_1,\ldots,j_l,j'_1,\ldots,j'_{l'}}-A_{j_1,\ldots,j_l}A_{j'_1,\ldots,j'_{l'}}\right)N^n + O(N^{n-2})
\end{split}
\end{displaymath}
where $n$, $n_f$ and $n_g$ are defined exactly as above in \eqref{eqn_corrdeg} and \eqref{eqn_corrdegs}.

As before, in the limit the $C_{i_1,\ldots,i_k}$ and $C_{i'_1,\ldots,i'_{k'}}$ terms cancel along with the factorials, leaving us with the expression \eqref{eqn_largeNcorrelation_3} for the large $N$ correlation. Note that both the numerator and the two factors in the denominator of \eqref{eqn_largeNcorrelation_3} are positive when $l,l'>0$. This follows from \eqref{eqn_chordcoeffdef}.

\item
Defining $n$, $n_f$ and $n_g$ according to \eqref{eqn_corrdeg} and \eqref{eqn_corrdegs} and applying Theorem \ref{thm_coeffoddeven} and Equation \eqref{eqn_subleadingsplit} we compute
\begin{displaymath}
\begin{split}
\cov{f_N}{g_N} =& C_{i_1,\ldots,i_k,i'_1,\ldots,i'_{k'}}A_{j_1,\ldots,j_l}N^n + O(N^{n-2}) \\
&-\big{(}C_{i_1,\ldots,i_k}A_{j_1,\ldots,j_l}N^{n_f} + O(N^{n_f-2})\big{)} \big{(}C_{i'_1,\ldots,i'_{k'}}N^{n_g} + O(N^{n_g-2})\big{)} \\
=& O(N^{n-2}) \\
\var{f_N} =& C_{i_1,\ldots,i_k}^2 \left(A_{j_1,\ldots,j_l,j_1,\ldots,j_l} - A_{j_1,\ldots,j_l}^2\right)N^{2n_f} + O(N^{2(n_f-1)}) \\
\var{g_N} =& C_{i'_1,\ldots,i'_{k'}}^2 N^{2n_g} + C_{i'_1,\ldots,i'_{k'},i'_1,\ldots,i'_{k'}}(1) N^{2(n_g-1)} + O(N^{2(n_g-2)}) \\
&- \left(C_{i'_1,\ldots,i'_{k'}}N^{n_g} + C_{i'_1,\ldots,i'_{k'}}(1)N^{n_g-2} + O(N^{n_g-4})\right)^2 \\
=& \left(C_{i'_1,\ldots,i'_{k'},i'_1,\ldots,i'_{k'}}(1) - 2C_{i'_1,\ldots,i'_{k'}}(1)C_{i'_1,\ldots,i'_{k'}}\right)N^{2(n_g-1)} + O(N^{2(n_g-2)}) \\
=& \left(\sum_{r,s=1}^{k'}\widetilde{C}_{i'_r,i'_s}(1)C_{i'_1,\ldots,\widehat{i'_r},\ldots,i'_{k'}}C_{i'_1,\ldots,\widehat{i'_s},\ldots,i'_{k'}}\right)N^{2(n_g-1)} + O(N^{2(n_g-2)}).
\end{split}
\end{displaymath}

Putting these together we get
\begin{multline*}
\corr{f_N}{g_N} = \\
\frac{O(\frac{1}{N})}{\sqrt{C_{i_1,\ldots,i_k}^2 \left(A_{j_1,\ldots,j_l,j_1,\ldots,j_l} - A_{j_1,\ldots,j_l}^2\right)\sum_{r,s=1}^{k'}\widetilde{C}_{i'_r,i'_s}(1)C_{i'_1,\ldots,\widehat{i'_r},\ldots,i'_{k'}}C_{i'_1,\ldots,\widehat{i'_s},\ldots,i'_{k'}} + O(\frac{1}{N^2})}}
\end{multline*}
from which \eqref{eqn_largeNcorrelation_4} follows.

\item
Proceeding in the same manner as above using Equation \eqref{eqn_subleadingsplit} we compute
\[ \cov{f_N}{g_N} = \left(\sum_{r=1}^k\sum_{s=1}^{k'}\widetilde{C}_{i_r,i'_s}(1)C_{i_1,\ldots,\widehat{i_r},\ldots,i_k}C_{i'_1,\ldots,\widehat{i'_s},\ldots,i'_{k'}}\right)N^{n-2}+O(N^{n-4}), \]
where $n$ is of course defined by \eqref{eqn_corrdeg}. From this we get
\begin{multline*}
\corr{f_N}{g_N} = \\
\frac{\sum_{r=1}^k\sum_{s=1}^{k'}\widetilde{C}_{i_r,i'_s}(1)C_{i_1,\ldots,\widehat{i_r},\ldots,i_k}C_{i'_1,\ldots,\widehat{i'_s},\ldots,i'_{k'}}+O(\frac{1}{N^2})}{\sqrt{\left(\sum_{r,s=1}^k\widetilde{C}_{i_r,i_s}(1)C_{i_1,\ldots,\widehat{i_r},\ldots,i_k}C_{i_1,\ldots,\widehat{i_s},\ldots,i_k}\right)\left(\sum_{r,s=1}^{k'}\widetilde{C}_{i'_r,i'_s}(1)C_{i'_1,\ldots,\widehat{i'_r},\ldots,i'_{k'}}C_{i'_1,\ldots,\widehat{i'_s},\ldots,i'_{k'}}\right)+O(\frac{1}{N^2})}}.
\end{multline*}
Taking the limit and using Equation \eqref{eqn_coeffsublead} we see that the coefficients $C_{i_1,\ldots,i_k}$ and $C_{i'_1,\ldots,i'_{k'}}$ in the numerator and the denominator cancel leaving us with the expression~\eqref{eqn_largeNcorrelation_5}.
\end{enumerate}
\end{proof}

\subsection{A generalization of Wigner's semicircle law and asymptotically free random variables}
\label{sec: voic}

As a final application of our results on the large $N$ asymptotic behavior of the correlation functions~\eqref{eqn_corfunmultitrace},
we will prove a generalization of Wigner's law to multi-trace functions.

\begin{theorem} \label{thm_wignergeneralize}
For any $m\geq 1$ and polynomials $q_1,\ldots,q_m\in\mathbb{C}[x]$,
\begin{equation} \label{eqn_wignergeneralize}
\lim_{N\to\infty}\left(\frac{\int_{\her{N}}\prod_{i=1}^m\left[\frac{1}{N}\tr\left(q_i\left(\frac{X}{\sqrt{N}}\right)\right)\right]e^{-\frac{1}{2}\tr(X^2)}\mathrm{d}X}{\int_{\her{N}}e^{-\frac{1}{2}\tr(X^2)}\mathrm{d}X}\right) = \prod_{i=1}^m\left[\frac{1}{2\pi}\int_{-2}^2 q_i(x)\sqrt{4-x^2}\mathrm{d}x\right].
\end{equation}
\end{theorem}

The case $m=1$ is Wigner's original semicircle law. In this sense Equation \eqref{eqn_wignergeneralize} may be restated as an equation amongst expectation values:
\[ \lim_{N\to\infty}\left\langle\prod_{i=1}^m\left[\frac{1}{N}\tr\left(q_i\left(\frac{X}{\sqrt{N}}\right)\right)\right]\right\rangle = \prod_{i=1}^m\left[\lim_{N\to\infty}\left\langle\frac{1}{N}\tr\left(q_i\left(\frac{X}{\sqrt{N}}\right)\right)\right\rangle\right]. \]
In other words, we can swap the limit with the product;
we may say that the expectation value of the product becomes the product of the expectation values.
Thus, in the large $N$ limit, we observe the vanishing of certain variances and covariances.
In particular, in the terminology of Voiculescu's free probability theory \cite{voiculimitlaw}, this result implies that the family of random variables
\[ \frac{1}{N}\tr\left(q\left(\frac{X}{\sqrt{N}}\right)\right), \quad X\in\her{N} \]
in the Gaussian Unitary Ensemble defined by polynomials $q\in\mathbb{C}[x]$ forms a family of \textit{asymptotically free} random variables.

\begin{proof}[Proof of Theorem \ref{thm_wignergeneralize}]
We begin by noting that both sides of \eqref{eqn_wignergeneralize} are multilinear functions of the polynomials $q_1,\ldots,q_m$. It is therefore sufficient to assume that $q_1,\ldots,q_m$ are monomials of the form~$q_r(x)=x^{k_r}$.

Obviously, if there is an odd number of indices $k_r$ having odd parity, then both sides of \eqref{eqn_wignergeneralize} will vanish before even taking the limit.
Therefore we may begin by assuming that the number of indices $k_r$ having odd parity is $2l>0$.
In this case the right-hand side of \eqref{eqn_wignergeneralize} clearly vanishes, whilst it follows from Theorem \ref{thm_degree} that the left-hand side is $O(\frac{1}{N^{2l}})$ and hence vanishes in the limit.

It remains to consider the case when $l=0$ so that all the indices $k_r=2i_r$ are even.
In this case it follows from Theorem~\ref{thm_degree} that the left-hand side is
\[ \lim_{N\to\infty}\left(C_{i_1,\ldots,i_m} + O\left(\frac{1}{N^2}\right)\right) = C_{i_1,\ldots,i_m} \]
and the result now follows from Theorem \ref{thm_catalanproduct} and the following integral expression for the Catalan numbers:
\[ C_i = \frac{1}{2\pi}\int_{-2}^2 x^{2i}\sqrt{4-x^2}\mathrm{d}x, \quad i\geq 0; \]
which is well-known.
\end{proof}

\begin{rem}
Using the results of Section \ref{sec_leadcoeff} and \ref{sec_subleadcoeff} we may compute the following large $N$ covariance, refining ever so slightly the above result.
Let $f,g\in\mathbb{C}[x]$ be two polynomials.
Then
\begin{multline} \label{eqn_largeNcov}
\lim_{N\to\infty}\left[\cov{\tr\left(f\left(\frac{X}{\sqrt{N}}\right)\right)}{\tr\left(g\left(\frac{X}{\sqrt{N}}\right)\right)}\right] = \\
\sum_{i,j=1}^\infty\frac{f^{(2i)}(0)g^{(2j)}(0)}{(i+j)i!(i-1)!j!(j-1)!} + \sum_{i,j=0}^{\infty}\frac{f^{(2i+1)}(0)g^{(2j+1)}(0)}{(i+j+1)(i!)^2(j!)^2}.
\end{multline}
We leave the details to the reader. If there is a more intrinsic expression for the quantity on the right-hand side of \eqref{eqn_largeNcov}, it is not readily apparent to the authors of this article.
\end{rem}

\end{document}